\documentclass[11pt,letterpaper]{article}

\usepackage[utf8]{inputenc}
\usepackage{amsmath,amssymb,amsfonts}
\usepackage{graphicx}
\usepackage[margin=3cm]{geometry}
\usepackage{revsymb}

\usepackage[american]{babel}
\usepackage{amsthm}

\theoremstyle{plain}
\newtheorem{theorem}{Theorem}[section]
\newtheorem{lemma}[theorem]{Lemma}

\theoremstyle{definition}
\newtheorem{definition}[theorem]{Definition}

\usepackage{tikz}

\usetikzlibrary{arrows.meta,babel,positioning,shapes.geometric,backgrounds}

\tikzset{
  > = Stealth,
  x = 1.5cm, y = 1.5cm,
  every node/.style = {
    draw,
    fill = white
  },
  plus/.style = {
    regular polygon,
    regular polygon sides = 5,
  },
  minus/.style = {
    regular polygon,
    regular polygon sides = 5,
    shape border rotate = 36,
    inner sep = 2pt,
  },
  zero/.style = {circle},
  hodge/.style = {rectangle},
  new/.style = {double},
}

\usepackage{authblk}

\usepackage{cite}
\usepackage{url}
\usepackage{hyperref}
\hypersetup{linktocpage,colorlinks=true,urlcolor=purple,linkcolor=purple,citecolor=purple}
\usepackage{xcolor}

\DeclareMathOperator{\diag}{diag}
\newcommand{\euler}{\ensuremath{\mathrm{e}}} 

\newcommand{\nwse}[3]{\ensuremath{#1^{#2}_{\phantom{#2} #3}}}
\newcommand{\swne}[3]{\ensuremath{#1_{#2}^{\phantom{#2} #3}}}
\newcommand{\rep}[3]{\left[ #1 \right]^{#2}_{\phantom{#2} #3}}
\newcommand{\riemann}[3][\mathring{\mathcal{R}}]{\nwse{#1}{#2}{#3}}

\newcommand{\hodge}{{\ast}} 
\newcommand{\extd}{\ensuremath{\mathrm{d}}}
\newcommand{\extD}{\ensuremath{\mathrm{D}}}
\newcommand{\SO}{\ensuremath{\mathrm{SO} \left( \eta_{+}, \eta_{-} \right)}}
\newcommand{\so}{\ensuremath{\mathfrak{so} \left( \eta_{+}, \eta_{-} \right)}}

\title{Wave Operators, Torsion, and Weitzenböck Identities}

\author[2,3,4]{José Barrientos\thanks{josebarrientos@udec.cl}}
\author[2]{Fernando Izaurieta\thanks{fizaurie@udec.cl}}
\author[5]{Eduardo Rodríguez\thanks{eduarodriguezsal@unal.edu.co}}
\author[1]{Omar Valdivia\thanks{ovaldivi@unap.cl}}

\affil[1]{Instituto de Ciencias Exactas y Naturales (ICEN), Facultad de Ciencias, Universidad Arturo Prat, 1110939 Iquique, Chile}
\affil[2]{Departamento de Física, Universidad de Concepción, Casilla 160-C, 4070105 Concepción, Chile}
\affil[3]{Departamento de Enseñanza de las Ciencias Básicas, Universidad Católica del Norte, Larrondo 1281, 1781421 Coquimbo, Chile}
\affil[4]{Institute of Mathematics of the Czech Academy of Sciences, Žitná 25, 11567 Praha 1, Czechia}
\affil[5]{Departamento de Física, Universidad Nacional de Colombia, 111321 Bogotá, Colombia}

\date{\today}

\begin{document}

\maketitle

\begin{abstract}
We offer a mathematical toolkit for the study of waves propagating on a background manifold with nonvanishing torsion.
Examples include electromagnetic and gravitational waves on a spacetime with torsion.
The toolkit comprises generalized versions of the Lichnerowicz--de~Rham and the Beltrami wave operators, and the Weitzenböck identity relating them on Riemann--Cartan geometries.
The construction applies to any field belonging to a matrix representation of a Lie (super) algebra containing an \so\ subalgebra.
Using these tools, we analyze the propagation of different massless waves in the eikonal (geometric optics) limit in a model-independent way and find that they all must propagate at the speed of light along null torsionless geodesics, in full agreement with the multimessenger observation GW170817/GRB170817A.
We also discuss how gravitational waves could be used as a probe to test for torsion.
\end{abstract}


\section{Physical motivation}
\label{sec:physmot}

On August 17, 2017, the detection of a binary neutron star merger heralded a new era of multi-messenger astronomy.
On the one hand, the LIGO/Virgo collaboration~\cite{PhysRevLett.119.161101}
announced the detection of gravitational wave (GW) signal GW170817,
while on the other hand,
Fermi and other observatories~\cite{Abbott_2017,Abbott_2017z,Goldstein_2017}
detected gamma-ray burst GRB170817A as its electromagnetic counterpart.
Remarkably, the almost simultaneous observations of both signals
set a strong limit of about one part in $10^{15}$
on the relative difference between the propagation speed of GWs
and the speed of light.
This event dramatically shows that comparison between speeds of waves propagating on a curved spacetime constrains the viable alternative theories of gravity, refuting many in the Horndeski family, among others~\cite{Ezquiaga:2018btd,Ezquiaga:2017ekz,Baker:2017hug,Sakstein:2017xjx,Heisenberg:2017qka,Kreisch:2017uet,Nojiri:2017hai}.

In this article, we provide a well-defined mathematical basis for studying the propagation of waves beyond the limits of Riemanniann geometry, that is to say, we introduce wave operators which act covariantly on gauge fields in a spacetime with nonvanishing torsion.

Our motivations are physical as well as mathematical.

From a physical and phenomenological point of view, the difficulty of studying the propagation of different kinds of waves on a background with nonvanishing torsion lies in the need to reformulate the whole baggage of dispersion relations, ray propagation, amplitudes and polarization in the context of Riemann--Cartan (RC) geometry and more general (super) symmetries.

In section~\ref{sec:mathmot} we delve in a little more detail into our mathematical motivations, which spur the definitions, lemmas, and theorems presented in sections~\ref{sec:gLdR} and~\ref{sec:opensuper}.

In section~\ref{Sec_EikonalTorsion}, we develop the eikonal (geometric optics) limit at leading and subleading order.
The first consequence we find is that regardless of torsion, the representation the field belongs to, and the corresponding generalized de~Rham operator, any massless field obeys the same canonical dispersion relation.
Therefore, all those fields propagate on null torsionless geodesics, regardless of whether the background torsion vanishes or not.
The point is subtle but crucial.
For instance, if GWs and electromagnetic waves (EMW) propagated both at the speed of light, but the former traveled along null auto-parallels, and the latter moved on torsionless null geodesics, we could have an unobserved delay among them even if their speeds coincide.
Based on this observation, we can conclude that timing observations like the ones of the multimessenger observation GW170817/GRB170817A do not refute the possibility of having a nonvanishing torsional background.

However, subleading order analysis of the eikonal limit reveals differences regarding the propagation of amplitude and polarization for different kinds of fields.
For instance, regardless of torsion, EMWs propagate in a vacuum obeying a wave equation in terms of the canonical de~Rham operator.
As a consequence, in the eikonal limit, their amplitudes propagates canonically (i.e., obeying the conservation of ray density), and their polarization is parallel transported along null torsionless geodesics.
In contrast, a field in a representation of the Lorentz group (e.g., a GW or a massless fermion) changes its behavior while it propagates along the same null torsionless geodesics.
The amplitude of such a field gets damped or amplified by torsion, and its polarization is no longer parallel transported along the trajectory; the background torsion scrambles the polarization modes while the wave propagates.

We conclude in section~\ref{sec:final} with some comments regarding whether GWs could be used as a probe to test for torsion.
The point is, of course, far from trivial.
There are many mechanisms which could give rise to torsion, ranging from the spin tensor in the Einstein--Cartan--Sciamma--Kibble (ECSK) theory~\cite{Kib61,Sciama:1964wt,Hehl:1971qi,doi:10.1142/6742,doi:10.1142/0356,Hehl76,Shapiro:2001rz,Hammond:2002rm,Poplawski:2009fb}
in the case of high density fermionic matter~\cite{Kerlick1975},
or dark matter with nonvanishing spin tensor~\cite{Izaurieta:2020xpk},
nonminimal couplings of the Horndeski kind~\cite{Valdivia:2017sat},
quadratic Lagrangians of the kind found in the Poincaré Gauge Theory~\cite{Hehl1980,Blagojevic:2013xpa},
couplings to topological invariants~\cite{Barrientos:2019awg,Alexander:2019wne,Magueijo:2019vmk,Barker:2020gcp,Toloza:2013wi,doi:10.1142/0224},
etc.

From a phenomenological point of view and regardless of its origins, detecting torsion appears extremely difficult.
Among Standard Model particles, fermions should interact with torsion~\cite{SupergravityVanProeyen,Chandia:1998nu}.
However, the effect is so weak
that it is hard to imagine any foreseeable particle physics experiment which might detect it~\cite{SupergravityVanProeyen}.
(There are no experimental reasons to expect them, but if there were any nonminimal couplings to standard model matter, the effects would be stronger and torsion easier to detect~\cite{Puetzfeld:2014sja}).
That is why it becomes important to analyze whether GWs could provide an alternative way to test for or even rule out the presence of torsion.


\section{Mathematical motivation}
\label{sec:mathmot}

From a theoretical point of view, it is not evident what the wave operator should be when considering a nonvanishing torsion geometry.
There are model-dependent approaches to GWs on RC geometries~\cite{Garcia:2000yi,Obukhov:2006gy,Obukhov:2017pxa,Jimenez-Cano:2020chm}.
The general idea is to read back the dispersion relation studying perturbations on a theory's equations of motion. Nevertheless, without a wave operator for these geometries, it is difficult to analyze subleading relations.
 
Instead of choosing a particular Lagrangian, here we focus, in a model-independent way, on how to define a wave operator for RC geometry and extending it to more general (super) symmetries. 

An important observation regarding the formulation of this problem is that in standard torsionless Riemannian geometry there are \emph{two} natural choices for the wave operator.
One is the de~Rham operator,
or ``the mathematician's Laplacian,''
and the other one is the Beltrami operator,
or ``the physicist's Laplacian.''
(See section~\ref{sec:basicdefs} for the formal definitions of all operators we discuss here).
Both operators are related to each other through
the so-called Weitzenböck identities~\cite{Bourguignon1990}.
In section~\ref{sec:gLdR}, we define and study some differential operators which are crucial to understanding, in a simple way, the underlying structure of these identities.
Furthermore, we construct generalizations of the de~Rham and Beltrami operators for RC geometry [and more general (super) symmetries] in such a way that they still satisfy a Weitzenböck identity.
These generalized wave operators have the remarkable property of being precisely the ones that arise from the Einstein--Hilbert term when considering GWs on a torsional background~\cite{Barrientos:2019awg,Valdivia:2017sat}.
We find that for fields belonging to different representations of a (super) algebra, we should also use different generalized de~Rham and Beltrami operators.
For instance, when torsion does not vanish, the wave operators acting on the electromagnetic field should be different from the wave operators acting on a gravitino or a GW.

The need for these new operators can be traced back to the difference between the Lorentz-covariant exterior derivative~\extD\ and the affine covariant derivative~$\nabla$.
Most GR textbooks define the affine covariant derivative%
\footnote{In this article we use in general a non-vanishing torsion,
but the metricity condition is always imposed, $\nabla_{\lambda}g_{\mu \nu}=0$.}
of a vector field $\vec{V}=V^{\mu}\partial_{\mu}$ as
\begin{equation}
  \nabla \vec{V} = \left(
    \partial_{\mu} V^{\nu} + \Gamma_{\mu \lambda}^{\nu} V^{\lambda}
  \right)
  \extd x^{\mu} \otimes \partial_{\nu},
\end{equation}
where the connection coefficients $\Gamma_{\mu \lambda}^{\nu}$ behave in such a
way as to ensure that the components $\nabla_{\mu}V^{\nu}$ transform as a
type-$\left( 1,1 \right)$ tensor under general coordinate transformations.
The canonical way to relate $\nabla$ and $\extD$ is to consider%
\begin{equation}
  \extD \vec{V} = \left(
    \partial_{\mu} V^{a} + \nwse{\omega}{a}{b\mu} V^{b}
  \right)
  \extd x^{\mu} \otimes \vec{e}_{a},
\end{equation}
and to demand $\nabla \vec{V} = \extD \vec{V}$ for all vector fields $\vec{V}$,
leading to the so-called \emph{vielbein postulate},
\begin{equation}
  \partial_{\mu} \nwse{e}{a}{\nu} + \nwse{\omega}{a}{b\mu} \nwse{e}{b}{\nu} -
  \Gamma_{\mu \nu}^{\lambda} \nwse{e}{a}{\lambda} = 0.
  \label{eq:VP}
\end{equation}
Eq.~\eqref{eq:VP} establishes a relation between the affine connection
$\Gamma_{\mu \nu}^{\lambda}$ and the Lorentz connection $\omega^{a}{}_{b\mu}$
that makes $\nabla$ and $\extD$ equivalent when they are applied to any
tensor zero-form field.

This is, however, as far as we can go.
The operators $\extD$ and $\nabla$ differ
when $p\geq1$ and torsion does not vanish.
Compare, for instance
$\extD^{2} \vec{V} = \extD \wedge \extD \vec{V}$ and
$\nabla \wedge \nabla \vec{V}$:
\begin{align}
  \extD^{2} \vec{V} & =
  \frac{1}{2} \nwse{R}{a}{b\mu \nu} V^{b}
  \extd x^{\mu} \wedge \extd x^{\nu} \otimes \vec{e}_{a},
  \label{Eq_Orig_D2} \\
  \nabla \wedge \nabla \vec{V} & =
  \frac{1}{2} \left( \nwse{\mathcal{R}}{\rho}{\sigma \mu \nu} V^{\sigma} -
  \nwse{T}{\lambda}{\mu \nu} \nabla_{\lambda} V^{\rho} \right)
  \extd x^{\mu} \wedge \extd x^{\nu} \otimes \partial_{\rho}.
  \label{Eq_Orig_Nabla2}
\end{align}
The difference between \extD\ and $\nabla$ when acting on $p$-forms can be best appreciated by distinguishing between coordinate-basis and orthonormal-basis indices.
The \extD\ operator is a covariant derivative in the principal-bundle sense, meaning that it does not perform parallel transport on the (usually omitted) $p$-form coordinate-basis indices.
The $\nabla$ operator is a covariant derivative in an affine sense, because it applies parallel transport on every index regardless of kind, including the $p$-form coordinate-basis indices.
Of course, the difference between \extD\ and $\nabla$ is only relevant when $T^{a} \neq 0$.

The de~Rham and Beltrami wave operators are related via the
\emph{Weitzenböck identity}, a relation of the form
\begin{equation}
  \Box_{\text{dR}} = \Box_{\text{B}} +
  \text{(Riemann curvature terms)}.
\end{equation}
The ``Riemann curvature terms'' change depending on the degree of the form that the wave operators act on.
The identity is well-known from the first half of the 20th century,
but cumbersome%
\footnote{The result also has a dark origin.
To the best of our knowledge,
the name ``Weitzenböck identity'' seems to be a misattribution,
as the identity never appears in the work of Austrian mathematician
Roland Weitzenböck (1885--1955).
If the reader has information about the real origin of the identity,
we would be glad to be contacted and to learn about it.}
to write down in a general form~\cite{choquet1977analysis,nla.cat-vn2659416}.

The de~Rham operator does not act covariantly on $p$-forms belonging to some representation of the Lorentz algebra, e.g., a gravitino field $\psi^{A} = \nwse{\psi}{A}{\mu} \extd x^{\mu}$.
To account for these cases, it is possible to generalize the de~Rham
operator to the Lichnerowicz--de~Rham operator.
The Lichnerowicz--de~Rham and Beltrami operators also satisfy a Weitzenböck
identity of the form
\begin{equation}
  \Box_{\text{LdR}} = \Box_{\text{B}} +
  \text{(even more Riemann curvature terms)}.
\end{equation}


Two examples of Weitzenböck identities for the de~Rham/Beltrami and the Lichnerowicz--de~Rham/Beltrami operators follow.

Let $A = A_{\mu} \extd x^{\mu}$ be the electromagnetic potential one-form in the Lorenz gauge, $\extd^{\dag} A = 0$.
The actions of the de~Rham and Beltrami wave operators on~$A$ are related through the Weitzenböck identity
\begin{equation}
  \Box_{\text{dR}} A = \Box_{\text{B}} A +
  \riemann{}{\mu \nu} A^{\mu} \extd x^{\nu},
  \label{Eq_Ej-W-1}
\end{equation}
where
$\riemann{}{\mu \nu} = \riemann{\lambda}{\mu \lambda \nu}$
is the Ricci tensor computed from the torsionless Riemann tensor.
For the field strength two-form
$F = \extd A = \frac{1}{2} F_{\mu\nu} \extd x^{\mu} \wedge \extd x^{\nu}$
we find
\begin{equation}
  \Box_{\text{dR}} F = \Box_{\text{B}} F + \frac{1}{2} \left(
    \riemann{}{\lambda\mu} \nwse{F}{\lambda}{\nu} -
    \riemann{}{\lambda\nu} \nwse{F}{\lambda}{\mu} -
    \riemann{\rho\sigma}{\mu\nu} F_{\rho\sigma}
  \right) \extd x^{\mu} \wedge \extd x^{\nu}.
  \label{Eq_Ej-W-2}
\end{equation}

For the Lichnerowicz--de~Rham operator, perhaps the best known example of a Weitzen\-böck identity concerns its action on the Riemann tensor itself,
\begin{align}
  \Box_{\text{LdR}} \mathring{R}^{ab} & =
  \Box_{\text{B}} \mathring{R}^{ab} +
  \frac{1}{2} \nwse{e}{a}{\alpha} \nwse{e}{b}{\beta} \left[
    \riemann{\lambda}{\mu} \riemann{\alpha\beta}{\lambda\nu} -
    \riemann{\lambda}{\nu} \riemann{\alpha\beta}{\lambda\mu} -
    \riemann{\alpha\beta}{\rho\sigma} \riemann{\rho\sigma}{\mu\nu} +
  \right. \nonumber \\ & \left.
    -2 \left(
      \riemann{\alpha}{\rho\sigma\mu} \riemann{\beta\rho\sigma}{\nu} -
      \riemann{\alpha}{\rho\sigma\nu} \riemann{\beta\rho\sigma}{\mu}
    \right)
  \right] \extd x^{\mu} \wedge \extd x^{\nu}.
  \label{Eq_Ej-W-3}
\end{align}

These examples should suffice to give the reader a taste for how complex those ``Riemann curvature terms'' can become.

In section~\ref{sec:newdefs}, we define a generalized Lichnerowicz--de~Rham operator that acts covariantly on $p$-form fields belonging to a representation of $\mathfrak{g}$ on a manifold with torsion that satisfies a Weitzenböck identity relating it to the Beltrami operator and provide a simple general recipe for the ``Riemann curvature terms'' that appear.
Needless to say, this generalization proves to be nontrivial.
Simpy generalizing the \so-covariant torsionless exterior coderivative (see definition~\ref{def:Ddag}) and using it to build a generalized Lichnerowicz--de~Rham operator following definition~\ref{def:LdRbox} fails to yield an operator that satisfies an acceptable Weitzenböck identity.
Instead, the final result depends on torsion, derivatives of torsion, and derivatives of the $p$-form on which the operator is applied.
Most importantly for us, this kind of operator does not arise in physical applications such as GWs on RC geometries.


\section{Generalized Gauge Lichnerowicz--de~Rham operator and its
Weitzenböck identity}
\label{sec:gLdR}


\subsection{Basic definitions}
\label{sec:basicdefs}

Let~$P$ be a principal bundle with gauge group~$G$ and basis manifold~$M$.
We shall concern ourselves only with the case where~$G$ contains the
indefinite special orthogonal group \SO\ as a subgroup
and $d = \eta_{+} + \eta_{-}$ matches the dimension of~$M$.


Associated with the \SO\ structure,
$M$ is endowed with a (pseudo) Riemannian metric tensor $g$
having $\eta_{+}$ positive and $\eta_{-}$ negative eigenvalues.
The most relevant physical applications are $\eta_{-} = 0$,
which corresponds to a Riemannian manifold with Euclidean signature,
and $\eta_{-} = 1$,
which we interepret as a spacetime manifold with Lorentzian signature.
The space of all $p$-forms on~$M$ is denoted by $\Omega^{p} \left( M \right)$.

In General Relativity (GR),
the standard treatment uses a set of local coordinates
$\left\{ x^{\mu} \right\}$ to define the local bases
$\left\{ \partial_{\mu} \right\}$ for the tangent space $T_{x} M$ and
$\left\{ \extd x^{\mu} \right\}$ for the cotangent space $T_{x}^{\hodge} M$.
We use lowercase Greek indices ($\lambda, \mu, \nu, \ldots$)
to refer to the coordinate-basis components of vectors and tensors on $M$. 
For instance, the metric tensor can be written in this basis as
$g = g_{\mu\nu} \extd x^{\mu} \otimes \extd x^{\nu}$
and its inverse as $g^{-1} = g^{\mu\nu} \partial_{\mu} \otimes \partial_{\nu}$.


Another usual choice of basis in~GR is the vielbein basis
$\left\{ \vec{e}_{a} \right\}$ for $T_{x} M$
and co-vielbein basis $\left\{ e^{a} \right\}$ for $T_{x}^{\hodge} M$.
These bases are (pseudo) orthonormal in the sense that the metric and its inverse
can be written as $g = \eta_{ab} e^{a} \otimes e^{b}$
and $g^{-1} = \eta^{ab} \vec{e}_{a} \otimes \vec{e}_{b}$,
respectively, where
$\eta_{ab} = \diag \left( -1, \ldots, -1, +1, \ldots, +1 \right)$,
with $\eta_{-}$ negative and $\eta_{+}$ positive values,
is an invariant tensor of \SO.
We use lowercase Latin indices ($a, b, c, \ldots$)
for the orthonormal-basis components of vectors and tensors on $M$.
The coordinate-basis components of
$e^{a} = \nwse{e}{a}{\mu} \extd x^{\mu}$ and
$\vec{e}_{a} = \swne{e}{a}{\mu} \partial_{\mu}$
allow us to map tensor components from the coordinate basis
to the orthonormal basis and vice~versa.

To prepare the ground for the introduction of the de~Rham and Beltrami wave operators, we begin with two well-known definitions.


\begin{definition}[Hodge star operator]
  \label{def:hodge}
  The \emph{Hodge star operator}~\cite{Fla89} is a linear map, $\hodge: \Omega^{p} \left( M \right) \to \Omega^{d-p} \left( M \right)$, that takes a differential $p$-form $\alpha \in \Omega^{p} \left( M \right)$,
  \begin{equation}
    \alpha = \frac{1}{p!} \alpha_{\mu_{1} \cdots \mu_{p}}
    \extd x^{\mu_{1}} \wedge \cdots \wedge \extd x^{\mu_{p}},
  \end{equation}
  and maps it into its \emph{Hodge dual}, $\hodge \alpha \in \Omega^{d-p} \left( M \right)$, defined by
  \begin{equation}
    \hodge \alpha =
    \frac{\sqrt{\left\vert g \right\vert}}{p! \left( d-p \right)!}
    \epsilon_{\mu_{1} \cdots \mu_{d}}
    \alpha^{\mu_{1} \cdots \mu_{p}}
    \extd x^{\mu_{p+1}} \wedge \cdots \wedge \extd x^{\mu_{d}},
  \end{equation}
  where $g = \det \left[ g_{\mu\nu} \right]$ is the determinant of the metric tensor and $\epsilon_{\mu_{1} \cdots \mu_{d}}$ is the totally antisymmetric Levi-Civita tensor density.
\end{definition}


\begin{definition}[Exterior coderivative]
  \label{def:ddag}
  The \emph{exterior coderivative} is a linear map, 
  $\extd^{\dag}: \Omega^{p} \left( M \right) \to \Omega^{p-1} \left( M \right)$,
  given by
  \begin{equation}
    \extd^{\dag} = -\left( -1 \right)^{d \left( p+1 \right) + \eta_{-}}
    \hodge \left( \extd \hodge \right.,
  \end{equation}
  where $\hodge$ stands for the Hodge star operator (see definition~\ref{def:hodge})
  and \extd\ is the usual exterior derivative operator.
\end{definition}


\begin{definition}[de~Rham operator]
  \label{def:dRbox}
  The \emph{de~Rham operator} is a linear map,
  $\Box_{\text{dR}}: \Omega^{p} \left( M \right) \to \Omega^{p} \left( M \right)$,
  defined as
  \begin{equation}
    \square_{\text{dR}} = \extd^{\dag} \extd + \extd \extd^{\dag},
    \label{eq:dRbox}
  \end{equation}
  where \extd\ and $\extd^{\dag}$ stand for, respectively,
  the exterior derivative and coderivative
  (see definition~\ref{def:ddag}) operators.
\end{definition}


\begin{definition}[Beltrami operator]
  \label{def:Bbox}
  The \emph{Beltrami operator} is a linear map,
  $\Box_{\text{B}}: \Omega^{p} \left( M \right) \to \Omega^{p} \left( M \right)$,
  given by
  \begin{equation}
    \Box_{\text{B}} = -\mathring{\nabla}^{\mu} \mathring{\nabla}_{\mu},
  \end{equation}
  where $\mathring{\nabla}_{\mu}$ stands for the
  torsionless affine covariant derivative
  $\mathring{\nabla} = \partial + \mathring{\Gamma}$
  defined in terms of the Christoffel connection
  $\mathring{\Gamma}^{\lambda}_{\mu\nu}$.
\end{definition}

We shall use a matrix representation~$\varphi$ for the Lie (super) algebra $\mathfrak{g}$ of the (super) group $G$ and the Lie algebra \so\ of \SO, with $\so \subset \mathfrak{g}$.
The matrix components of $\varphi \left( X \right)$, with $X \in \mathfrak{g}$, are written simply as $\nwse{X}{M}{N}$, using uppercase Latin indices ($L, M, N, \ldots$).

Let us denote the matrix components of the projection of the $1$-form connection on the basis manifold $M$ by
$\nwse{A}{M}{N} = \nwse{A}{M}{N\mu} \extd x^{\mu} \in \Omega^{1} \left( M \right)$.
In physical applications, it is usual to write the
\so-valued piece of the projection of the connection as
\begin{equation}
  \nwse{\omega}{M}{N} = \frac{1}{2} \omega^{ab} \rep{J_{ab}}{M}{N},
\end{equation}
where $J_{ab}=-J_{ba}$ label the $d \left( d-1 \right)/2$ generators of \so\ and
$\omega^{ab} = \nwse{\omega}{ab}{\mu} \extd x^{\mu} \in \Omega^{1} \left( M \right)$ are the components of the \so\ connection,
which we affectionately refer to as ``spin connection.''
In this way, the projected connection can be written in general as
\begin{align}
  \nwse{A}{M}{N} & = \nwse{\omega}{M}{N} + \nwse{a}{M}{N},
  \label{Eq_A=w+a}
  \\ & =
  \frac{1}{2} \omega^{ab} \rep{J_{ab}}{M}{N} + \nwse{a}{M}{N}.
  \label{Eq_A=wJ+a}
\end{align}


\begin{definition}[Covariant exterior derivatives]
  \label{def:covextders}
  The \emph{$\mathfrak{g}$-covariant exterior derivative} operator $\mathbb{D}$
  is a linear map that acts on a $p$-form to yield a $\left( p+1 \right)$-form,
  $\mathbb{D}: \Omega^{p}   \left( M \right) \to
               \Omega^{p+1} \left( M \right)$.
  Let $\psi^{M}$ be a $p$-form field belonging to the
  $\varphi$ representation of~$\mathfrak{g}$.
  The projected action of $\mathbb{D}$ on $\psi^{M}$ is defined as
  \begin{equation}
    \mathbb{D} \psi^{M} = \extd \psi^{M} + \nwse{A}{M}{N} \wedge \psi^{N},
  \end{equation}
  where \extd\ denotes the usual exterior derivative.
  The \so-covariant piece of $\mathbb{D} \psi^{M}$ is written as
  \begin{equation}
    \extD \psi^{M} = \extd \psi^{M} +
    \frac{1}{2} \omega^{ab} \rep{J_{ab}}{M}{N} \psi^{N}.
  \end{equation}
\end{definition}


\begin{definition}[Torsion and curvature]
  \label{def:torsioncurvature}
  The co-vielbein $e^{a}$ belongs to the vector representation of \so.
  Its \so-covariant exterior derivative defines the components of the
  \emph{torsion} two-form $T^{a} \in \Omega^{2} \left( M \right)$,
  \begin{equation}
    T^{a} = \extD e^{a} = \extd e^{a} + \nwse{\omega}{a}{b} \wedge e^{b}.
  \end{equation}
  The \so\ curvature two-form $R^{ab} \in \Omega^{2} \left( M \right)$ is defined in terms of the spin connection as
  \begin{equation}
    R^{ab} = \extd \omega^{ab} + \nwse{\omega}{a}{c} \wedge \omega^{cb}.
  \end{equation}
\end{definition}


In order to separate the torsional degrees of freedom from the metric ones, it
proves useful to define a purely metric, torsionless spin connection
$\mathring{\omega}^{ab}$ implicitly given by%
\footnote{We use the notation $\mathring{X}$
to denote the ``torsionless version'' of $X$.}
\begin{equation}
  \extd e^{a} + \nwse{\mathring{\omega}}{a}{b} \wedge e^{b} = 0.
\end{equation}
The difference between the full spin connection $\omega^{ab}$ and its
torsionless relative $\mathring{\omega}^{ab}$ defines the \emph{contorsion} tensor one-form,
\begin{equation}
  \kappa^{ab} = \omega^{ab} - \mathring{\omega}^{ab}.
  \label{eq:kappa}
\end{equation}
Contorsion is algebraically related to torsion via
$T^{a} = \nwse{\kappa}{a}{b} \wedge e^{b}$.
The curvature tensors for the full and torsionless spin connections are related by
\begin{equation}
  R^{ab} = \mathring{R}^{ab} + \mathring{\extD} \kappa^{ab} +
  \nwse{\kappa}{a}{c} \wedge \kappa^{cb},
  \label{Eq_Lorentz_Curvature}
\end{equation}
where $\mathring{\extD}$ stands for the \so-covariant exterior derivative
computed with the torsionless spin connection $\mathring{\omega}^{ab}$,
$\mathring{\extD} = \extd + \mathring{\omega}$
(see definition~\ref{def:covextders}).

Writing the coordinate components of $\mathring{R}^{ab}$ as
\begin{equation}
  \mathring{R}^{ab} = \frac{1}{2} \nwse{e}{a}{\rho} \nwse{e}{b}{\sigma}
  \riemann{\rho \sigma}{\mu \nu}
  \extd x^{\mu} \wedge \extd x^{\nu},
\end{equation}
it is a simple matter to prove that
$\riemann{\alpha \beta}{\mu \nu}$
correspond to the coordinate-basis components of the Riemann tensor.


\begin{definition}[\so-covariant torsionless exterior coderivative]
  \label{def:Ddag}
  The \emph{\so-co\-variant torsionless exterior coderivative} is a linear map,
  $\mathring{\extD}^{\dag}:
  \Omega^{p} \left( M \right) \to \Omega^{p-1} \left( M \right)$,
  given by
  \begin{equation}
    \mathring{\extD}^{\dag} = -\left( -1 \right)^{d\left(  p+1\right)+\eta_{-}}
    \hodge \left( \mathring{\extD} \hodge \right.,
  \end{equation}
  where $\hodge$ stands for the Hodge star operator (see definition~\ref{def:hodge})
  and $\mathring{\extD}$ is the \so-covariant torsionless exterior derivative operator, $\mathring{\extD} = \extd + \mathring{\omega}$ (see definition~\ref{def:covextders}).
\end{definition}


\begin{definition}[Lichnerowicz--de~Rham operator]
  \label{def:LdRbox}
  The \emph{Lichnerowicz--de~Rham operator} is a linear map,
  $\Box_{\text{LdR}}: \Omega^{p} \left( M \right) \to \Omega^{p} \left( M \right)$,
  defined as
  \begin{equation}
    \Box_{\text{LdR}} =
    \mathring{\extD}^{\dag} \mathring{\extD} +
    \mathring{\extD} \mathring{\extD}^{\dag},
  \end{equation}
  where $\mathring{\extD}$ and $\mathring{\extD}^{\dag}$ stand for, respectively,
  the \so-covariant torsionless exterior
  derivative (see definition~\ref{def:covextders}) and
  coderivative (see definition~\ref{def:Ddag}) operators.
\end{definition}


Figure~\ref{fig:old} shows the logical structure behind the operators defined in this section.


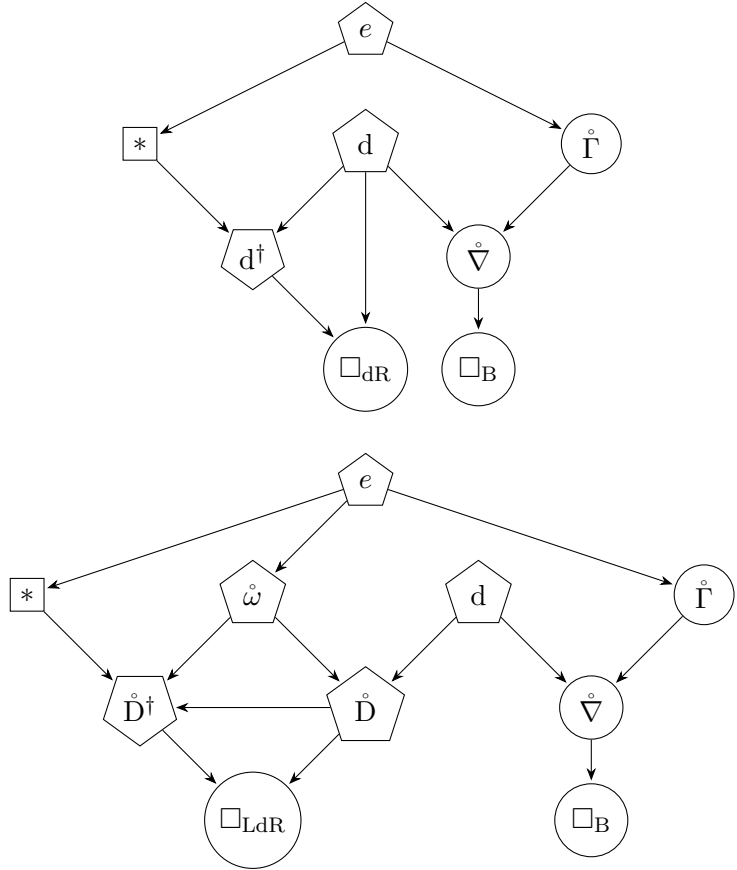
\begin{figure}
  \centering
  \begin{tikzpicture}
    \begin{scope}[shift = {(-1, 4)}]
      \node [plus]  (e)     at (4, 3) {$e$};
      \node [hodge] (hodge) at (2, 2) {$\hodge$};
      \node [plus]  (d)     at (4, 2) {$\extd$};
      \node [zero]  (Gammo) at (6, 2) {$\mathring{\Gamma}$};
      \node [minus] (ddag)  at (3, 1) {$\extd^{\dag}$};
      \node [zero]  (nablo) at (5, 1) {$\mathring{\nabla}$};
      \node [zero]  (dR)    at (4, 0) {$\Box_{\text{dR}}$};
      \node [zero]  (B)     at (5, 0) {$\Box_{\text{B}}$};
      \foreach \A / \B in {
        e/hodge, e/Gammo,
        hodge/ddag, d/ddag, d/nablo, Gammo/nablo,
        ddag/dR, d/dR, nablo/B} {
        \begin{scope}[on background layer]
          \draw [->, shorten > = 1pt] (\A) to (\B);
        \end{scope}
      }
    \end{scope}
    \node [plus]  (e)     at (3, 3) {$e$};
    \node [hodge] (hodge) at (0, 2) {$\hodge$};
    \node [plus]  (omego) at (2, 2) {$\mathring{\omega}$};
    \node [plus]  (d)     at (4, 2) {$\extd$};
    \node [zero]  (Gammo) at (6, 2) {$\mathring{\Gamma}$};
    \node [minus] (Dodag) at (1, 1) {$\mathring{\extD}^{\dag}$};
    \node [plus]  (Do)    at (3, 1) {$\mathring{\extD}$};
    \node [zero]  (nablo) at (5, 1) {$\mathring{\nabla}$};
    \node [zero]  (LdR)   at (2, 0) {$\Box_{\text{LdR}}$};
    \node [zero]  (B)     at (5, 0) {$\Box_{\text{B}}$};
    \foreach \A / \B in {
      e/hodge, e/omego, e/Gammo,
      hodge/Dodag, omego/Dodag, omego/Do, d/Do, d/nablo, Gammo/nablo, Do/Dodag,
      Dodag/LdR, Do/LdR, nablo/B} {
      \begin{scope}[on background layer]
        \draw [->, shorten > = 1pt] (\A) to (\B);
      \end{scope}
    }
  \end{tikzpicture}
  \caption{When acting on a scalar field, the de~Rham and Beltrami wave operators are defined in terms of the derivatives and coderivatives shown in the diagram above. When acting on a field in a representation of the Lorentz group on a torsionless manifold, on the other hand, the de~Rham operator must be generalized to the Lichnerowicz--de~Rham operator, as shown in the diagram below.
  Upwards and downwards pointing pentagons denote operators that respectively increase or decrease the rank of a differential form. Circles are used for operators that do not change the rank of a form, while the Hodge star operator is indicated with a rectangle.}
  \label{fig:old}
\end{figure}


\subsection{The I operator and its offspring}
\label{sec:newdefs}

As in section~\ref{sec:basicdefs}, the definitions we give here refer to a $d$-dimensional (pseudo) Riemannian manifold~$M$ which serves as the basis manifold for a principal bundle with gauge (super) group $G \supset \SO$ and $\eta_{+} + \eta_{-} = d$.

Definition~\ref{def:I} below provides the foundation for a generalized covariant coderivative, which then allows us to define generalized Beltrami and Lichnerowicz--de~Rham wave operators.


\begin{definition}[The $\mathrm{I}_{\theta}$ operator]
  \label{def:I}
  Let $\theta \in \Omega^{r} \left( M \right)$ be an $r$-form field on $M$.
  The \emph{$\mathrm{I}_{\theta}$ operator} is a linear map,
  $\mathrm{I}_{\theta}: \Omega^{p} \left( M \right)
   \to \Omega^{p-r} \left( M \right)$, defined by
  \begin{equation}
    \mathrm{I}_{\theta} =
    \left( -1 \right)^{\left( d-p \right) \left( p-r \right) + \eta_{-}}
    \hodge \left( \theta \wedge \hodge \right.,
  \end{equation}
  where $\hodge$ stands for the Hodge star operator (see definition~\ref{def:hodge}).
  When the $r$-form field $\theta$ corresponds to the basis element
  $e^{a_{1}} \wedge \cdots \wedge e^{a_{r}}$, we use the shorthand notation
  \begin{equation}
    \mathrm{I}^{a_{1}\cdots a_{r}} =
    \mathrm{I}_{e^{a_{1}}\wedge \cdots \wedge e^{a_{r}}}.
  \end{equation}
  For instance, for $r=1$ we write
  \begin{equation}
    \mathrm{I}^{a} = \mathrm{I}_{e^{a}} =
    \left( -1 \right)^{d \left( p-1 \right) + \eta_{-}}
    \hodge \left( e^{a} \wedge \hodge \right..
    \label{Eq_Ia}
\end{equation}
\end{definition}


The operators introduced in definition~\ref{def:I} satisfy a host of interesting properties.
For instance, let $\alpha$ be a $p$-form and $\beta$ a $q$-form on $M$.
When $r=0$, we have $\mathrm{I}_{\theta} \alpha = \theta \alpha$.
When $r=1$, $\mathrm{I}_{\theta}$ satisfies a sign-corrected Leibniz rule,
\begin{equation}
  \mathrm{I}_{\theta} \left( \alpha \wedge \beta \right) =
  \mathrm{I}_{\theta} \alpha \wedge \beta + 
  \left( -1 \right)^{p}
  \alpha \wedge \mathrm{I}_{\theta} \beta.
\end{equation}

We also have the algebra
$\mathrm{I}_{\beta} \mathrm{I}_{\alpha} = \mathrm{I}_{\alpha \wedge \beta}$,
which, in particular, implies that
$\mathrm{I}_{\theta}^{2} = 0$ for all odd-forms $\theta$.
A couple of direct consequences are
$\mathrm{I}_{b_{1}\cdots b_{n}} \mathrm{I}_{a_{1}\cdots a_{m}} =
 \mathrm{I}_{a_{1}\cdots a_{m} b_{1}\cdots b_{n}}$ and
$\mathrm{I}^{a} \mathrm{I}_{a} = 0$.

The operator $\mathrm{I}_{\theta}$ is useful to define a generalized Lie derivative and a generalized covariant coderivative.




\begin{definition}[$G$-torsion]
  \label{def:G-Torsion}
  The \emph{$G$-torsion} two-form $\mathbb{T}^{a}$ is defined in general as
  \begin{equation}
    \mathbb{T}^{a} = \mathbb{D} e^{a},
  \end{equation}
  where $e^{a}$ is the co-vielbein one-form and $\mathbb{D}$ stands for the $G$-covariant exterior derivative introduced in definition~\ref{def:covextders}.
  The \SO\ case corresponds to the standard torsion two-form,
  $T^{a} = \extD e^{a}$ (see definition~\ref{def:torsioncurvature}).
  
  Splitting the gauge connection as in eq.~(\ref{Eq_A=w+a}), we may write
  $\mathbb{T}^{a} = T^{a} + \nwse{a}{a}{b} \wedge e^{b}$,
  where the lowercase indices denote the associated representation of $e^{a}$.
  (It may seem strange to have the same matrix index
  acting on two pieces of the same Lie algebra.
  There are many examples of this kind of (super) algebras~\cite{Salgado:2014iha,Fierro:2014lka,Salgado:2014jka,Szabo:2014zua,Diaz:2012zza,Izaurieta:2011fr,Izaurieta:2009hz,Izaurieta:2006aj}).
\end{definition}


The purely affine (i.e., nonmetric) degrees of freedom contained in the $G$-torsion two-form can be expressed by means of the $G$-contorsion one-form $\mathbb{K}^{ab}$ in a way that is completely analogous to the \SO\ case.
Splitting the spin connection one-form into a torsionless piece plus the contorsion [cf.~ec.~(\ref{eq:kappa})], it is possible to write the $G$-torsion as
\begin{equation}
  \mathbb{T}^{a} = \nwse{\mathbb{K}}{a}{b} \wedge e^{b},
  \label{Eq_T=Ke}
\end{equation}
with
\begin{equation}
  \nwse{\mathbb{K}}{a}{b} = \nwse{\kappa}{a}{b} + \nwse{a}{a}{b}.
  \label{Eq_G-contorsion}
\end{equation}
Eq.~(\ref{Eq_T=Ke}) can be inverted to yield
\begin{equation}
  \mathbb{K}_{ab} = \frac{1}{2} \left(
    \mathrm{I}_{a} \mathbb{T}_{b} -
    \mathrm{I}_{b} \mathbb{T}_{a} +
    e^{c} \mathrm{I}_{ab} \mathbb{T}_{c}
  \right).
\label{Eq_K=K(T)}
\end{equation}


\begin{definition}[Generalized Lie derivative]
  \label{def:D_Lie}
  Let $\theta \in \Omega^{r} \left( M \right)$ be an $r$-form field on $M$.
  The \emph{generalized Lie derivative}
  $\mathfrak{D}_{\theta}:
   \Omega^{p} \left( M \right) \to \Omega^{p+1-r}\left( M \right)$
  associated to $\theta$ is a linear map defined by
  \begin{equation}
    \mathfrak{D}_{\theta} =
    \mathrm{I}_{\theta} \mathbb{D} +
    \mathbb{D}\mathrm{I}_{\theta},
    \label{Eq_Def_D_theta}
  \end{equation}
  where $I_{\theta}$ is the operator introduced in definition~\ref{def:I}
  and $\mathbb{D}$ stands for the $\mathfrak{g}$-covariant exterior derivative operator (see definition~\ref{def:covextders}).
  When the $r$-form field $\theta$ matches the basis element
  $e^{a_{1}} \wedge \cdots \wedge e^{a_{r}}$,
  we use the shorthand notation
  \begin{equation}
    \mathfrak{D}_{a_{1} \cdots a_{r}} =
    \mathrm{I}_{a_{1}\cdots a_{r}} \mathbb{D} +
    \mathbb{D} \mathrm{I}_{a_{1}\cdots a_{r}}.
  \end{equation}
  When we restrict ourselves to the $\SO \subset G$ case,
  we write $\mathcal{D}$ instead of $\mathfrak{D}$,
   \begin{equation}
     \mathcal{D}_{a_{1} \cdots a_{r}} =
     \mathrm{I}_{a_{1} \cdots a_{r}} \extD +
     \extD \mathrm{I}_{a_{1} \cdots a_{r}},
   \end{equation}
   where $\extD$ stands for the \so-covariant exterior derivative operator
   (see definition~\ref{def:covextders}).
\end{definition}


\begin{definition}[Generalized covariant coderivative]
  \label{def:Cov-Cod}
  The \emph{generalized covariant coderivative},
  $\mathbb{D}^{\ddag}:
  \Omega^{p} \left( M \right) \to \Omega^{p-1} \left( M \right)$,
  is a linear map given by
  \begin{equation}
    \mathbb{D}^{\ddag} = -\mathrm{I}_{a} \mathbb{D} \mathrm{I}^{a},
  \end{equation}
  where $I_{\theta}$ is the operator introduced in definition~\ref{def:I}
  and $\mathbb{D}$ stands for the $\mathfrak{g}$-covariant exterior derivative operator (see definition~\ref{def:covextders}).
  When we restrict ourselves to the $\SO \subset G$ case,
  we write
  \begin{equation}
    \extD^{\ddag} = -\mathrm{I}_{a} \extD \mathrm{I}^{a},
  \end{equation}
  where $\extD$ stands for the \so-covariant exterior derivative operator
  (see definition~\ref{def:covextders}).
\end{definition}


Beyond the fact that it lowers the degree of the differential form upon which it acts by one unit, calling $\mathbb{D}^{\ddag}$ a coderivative is justified in the sense that it differs from the naïve generalization (cf.~def.~\ref{def:Ddag})
$\mathbb{D}^{\dag} =
 -\left( -1 \right)^{d \left( p+1 \right) + \eta_{-}}
 \hodge \mathbb{D} \hodge$
only by carefully chosen terms related to the $G$-torsion two-form (cf.~def.~\ref{def:G-Torsion}).
Indeed, it is possible to show that its action on a $p$-form $\Phi^{A}$ can be written as
\begin{equation}
  \mathbb{D}^{\ddag}\Phi^{A}=\mathbb{D}^{\dag}\Phi^{A}-\frac{1}{2}\left[\mathrm{I}^{bc}\mathbb{T}^{a}\mathrm{I}_{bc}\left(  e_{a}\wedge \Phi^{A}\right)  +\left(  \mathrm{I}^{a}\mathbb{T}^{b}-\mathrm{I}^{b}\mathbb{T}^{a}\right)  \wedge \mathrm{I}_{ab}\Phi^{A}\right] 
  \label{Eq_CoDerivS}
\end{equation}

The following lemmas explore some of the many useful properties of the generalized Lie derivative $\mathfrak{D}_{\theta}$ introduced in definition~\ref{def:D_Lie}.
Their proofs are all straightforward and therefore omitted. 


\begin{lemma}
  \label{lem:LieLeibniz}
  When $\theta$ is a one-form,
  the generalized Lie derivative $\mathfrak{D}_{\theta}$
  (see definition~\ref{def:D_Lie})
  satisfies the Leibniz rule without any sign correction,
  \begin{equation}
    \mathfrak{D}_{\theta} \left( P \wedge Q \right) =
    \mathfrak{D}_{\theta} P \wedge Q +
    P \wedge \mathfrak{D}_{\theta} Q.
  \end{equation}
\end{lemma}


\begin{lemma}
  \label{lem:Lieinverse}
  When $\theta$ matches the co-vielbein one-form $e^{a}$,
  the defining relation for the generalized Lie derivative,
  $\mathfrak{D}_{a} = \mathrm{I}_{a} \mathbb{D} + \mathbb{D} \mathrm{I}_{a}$
  (see definition~\ref{def:D_Lie}),
  can be inverted to yield
  \begin{equation}
    \mathbb{D} = e^{a} \wedge \mathfrak{D}_{a} -
    \mathbb{T}^{a} \wedge \mathrm{I}_{a}.
  \end{equation}
\end{lemma}






\begin{lemma}
  \label{lem:calD=nabla}
  When $\theta$ matches the co-vielbein one-form $e^{a}$,
  the \so-covariant generalized Lie derivative $\mathcal{D}_{a}$
  (see definition~\ref{def:D_Lie}) can be written as
  \begin{equation}
    \mathcal{D}_{a} = \swne{e}{a}{\mu} \nabla_{\mu} +
    \mathrm{I}_{a} T^{b} \wedge \mathrm{I}_{b},
    \label{Eq_LieD=Nabla+T}
  \end{equation}
  where $T^{a}$ is the torsion two-form
  and $\nabla = \partial + \Gamma$
  stands for the usual RC covariant derivative
  with a general affine connection $\Gamma$.
  
  In particular, eq.~(\ref{Eq_LieD=Nabla+T}) shows that, when torsion vanishes,
  $\mathring{\mathcal{D}}_{a}$ and $\mathring{\nabla}_{\mu}$
  \emph{correspond to the same operation}.
\end{lemma}


Lemma~\ref{lem:calD=nabla} is a crucial result that serves both as inspiration for the following generalizations of the Beltrami and Lichnerowicz--de~Rham wave operators, and as the key to clarify the underlying algebraic structure of the Weitzenböck identities.


\begin{definition}[Generalized Beltrami operator]
  \label{def:Beltrami_op}
  The \emph{generalized Beltrami operator} is a linear map,
  $\blacksquare_{\text{B}}:
   \Omega^{p} \left( M \right) \to \Omega^{p} \left( M \right)$,
   defined as
   \begin{equation}
     \blacksquare_{\text{B}} = -\mathfrak{D}^{a} \mathfrak{D}_{a},
   \label{Eq_Gen_Beltrami}
   \end{equation}
   where $\mathfrak{D}_{a}$ stands for the generalized Lie derivative
   introduced in definition~\ref{def:D_Lie}.
\end{definition}


\begin{definition}[Generalized Lichnerowicz--de~Rham operator]
  \label{def:Lichnerowicz_op}
  The \emph{generalized Lichnerow\-icz--de~Rham operator} is a linear map,
  $\blacksquare_{\text{LdR}}:
   \Omega^{p} \left( M \right) \to \Omega^{p} \left( M \right)$,
  defined as
  \begin{equation}
    \blacksquare_{\text{LdR}} =
    \mathbb{D}^{\ddag} \mathbb{D} + \mathbb{DD}^{\ddag},
    \label{Eq_Def_LLdR_Generalisimo}
  \end{equation}
  where $\mathbb{D}$ and $\mathbb{D}^{\ddag}$ stand for the
  $\mathfrak{g}$-covariant exterior derivative and generalized coderivative
  introduced in definitions~\ref{def:covextders} and~\ref{def:Cov-Cod}.
\end{definition}


Figure~\ref{fig:new} shows the logical structure behind the operators defined in this section.


\begin{figure}
  \centering
  \begin{tikzpicture}
    \foreach \dy / \Wlab / \Dlab / \Ddaglab / \Lielab in {
      0 / $\omega$ / $\extD$      / $\extD^{\ddag}$      / $\mathcal{D}$,
      4 / $A$      / $\mathbb{D}$ / $\mathbb{D}^{\ddag}$ / $\mathfrak{D}$
    } {
      \begin{scope}[shift = {(0, \dy)}]
        \node [plus]  (d)     at (0, 3) {$\extd$};
        \node [plus]  (W)     at (2, 3) {\Wlab};
        \node [plus]  (e)     at (3, 3) {$e$};
        \node [hodge] (hodge) at (5, 3) {$\hodge$};
        \node [plus]  (D)     at (1, 2) {\Dlab};
        \node [minus, new] (I)    at (4, 2) {$\mathrm{I}$};
        \node [minus, new] (Ddag) at (2, 1) {\Ddaglab};
        \node [zero, new]  (lie)  at (3, 1) {\Lielab};
        \node [zero, new]  (bLdR) at (1, 0) {$\blacksquare_{\text{LdR}}$};
        \node [zero, new]  (bB)   at (3, 0) {$\blacksquare_{\text{B}}$};
        \foreach \A / \B in {
          e/hodge, d/D, W/D, e/I, hodge/I,
          D/Ddag, D/lie, I/Ddag, I/lie,
          D/bLdR, Ddag/bLdR, lie/bB} {
          \begin{scope}[on background layer]
            \draw [->, shorten > = 1pt] (\A) to (\B);
          \end{scope}
        }
      \end{scope}
    }
  \end{tikzpicture}
  \caption{On a manifold with torsion, both the (Lichnerowicz--) de~Rham and the Beltrami wave operators must be generalized to act covariantly on fields that carry a representation of the \SO\ group (above) or a more general (super) group (below).
  Upwards and downwards pointing pentagons denote operators that respectively increase or decrease the rank of a differential form. Circles are used for operators that do not change the rank of a form, while the Hodge star operator is indicated with a rectangle.
  New operators, introduced either here or in Refs.~\cite{Valdivia:2017sat,Barrientos:2019awg}, are highlighted by means of a double line.}
  \label{fig:new}
\end{figure}
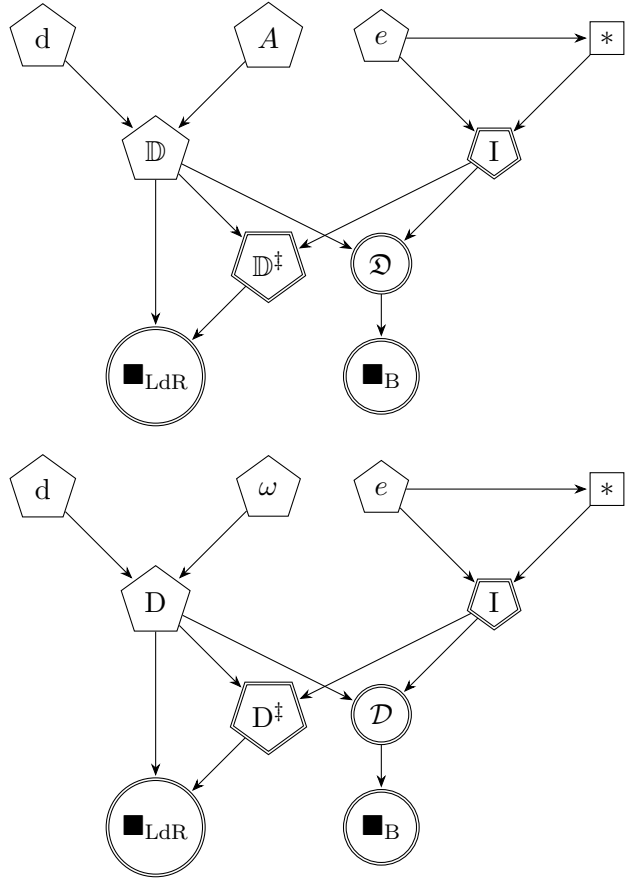


The generalized Beltrami and Lichnerowicz--de~Rham operators introduced in definitions~\ref{def:Beltrami_op} and~\ref{def:Lichnerowicz_op} are the protagonists of Theorem~\ref{th:Gen_Weitzenboeck}.


\begin{theorem}[Generalized Weitzenböck identity]
  \label{th:Gen_Weitzenboeck}
  The generalized Beltrami and Lichnerow\-icz--de~Rham operators introduced in definitions~\ref{def:Beltrami_op} and~\ref{def:Lichnerowicz_op} satisfy the following \emph{generalized Weitzenböck identity}:
  \begin{equation}
    \blacksquare_{\mathrm{LdR}} = \blacksquare_{\mathrm{B}} +
    \mathrm{I}_{a} \mathbb{D}^{2} \mathrm{I}^{a},
    \label{eq:Gen_Weitzenboeck}
  \end{equation}
  where $\mathrm{I}_{a}$ is the operator introduced in definition~\ref{def:I}
  and $\mathbb{D}$ stands for the $\mathfrak{g}$-covariant exterior derivative
  (see definition~\ref{def:covextders}).
  The second term on the right-hand side of eq.~(\ref{eq:Gen_Weitzenboeck})
  introduces terms proportional to the gauge curvature two-form,
  $F = \extd A + \frac{1}{2} \left[ A, A \right]$,
  via the Bianchi identity.
\end{theorem}

\begin{proof}
Since (cf.~def.~\ref{def:Cov-Cod})
$\mathbb{D}^{\ddag} = -\mathrm{I}_{a} \mathbb{D} \mathrm{I}^{a}$
and (cf.~def.~\ref{def:D_Lie})
$\mathfrak{D}_{a} = \mathrm{I}_{a} \mathbb{D} + \mathbb{D} \mathrm{I}_{a}$,
we may write
\begin{align*}
  \mathbb{D}^{\ddag} \mathbb{D} & =
  -\mathrm{I}_{a} \mathbb{D} \mathrm{I}^{a} \mathbb{D}
  \\ & =
  -\mathrm{I}_{a} \mathbb{D} \left(
    \mathfrak{D}^{a} - \mathbb{D} \mathrm{I}^{a}
  \right)
  \\ & =
  -\left( \mathfrak{D}_{a} - \mathbb{D} \mathrm{I}_{a} \right) \mathfrak{D}^{a} +
  \mathrm{I}_{a} \mathbb{D}^{2} \mathrm{I}^{a}
  \\ & =
  -\mathfrak{D}_{a} \mathfrak{D}^{a} +
  \mathbb{D} \mathrm{I}_{a} \left(
    \mathrm{I}^{a} \mathbb{D} + \mathbb{D} \mathrm{I}^{a}
  \right)
  +\mathrm{I}_{a} \mathbb{D}^{2} \mathrm{I}^{a}.
\end{align*}
Recalling that $\mathrm{I}_{a}\mathrm{I}^{a}=0$, we find
$\mathbb{D}^{\ddag} \mathbb{D} =
-\mathfrak{D}_{a} \mathfrak{D}^{a} -
 \mathbb{DD}^{\ddag} +
 \mathrm{I}_{a} \mathbb{D}^{2} \mathrm{I}^{a}$,
whence we get
\begin{equation}
  \mathbb{D}^{\ddag} \mathbb{D} + \mathbb{DD}^{\ddag} =
  -\mathfrak{D}_{a} \mathfrak{D}^{a} + \mathrm{I}_{a} \mathbb{D}^{2} \mathrm{I}^{a}.
  \label{Eq_Weitzenboeck_Generalisimo}
\end{equation}
\end{proof}


Before commenting on the features of the general case, let us observe that the
standard Weitzenböck identity for the de~Rham operator and its generalization for the Lichnerowicz--de~Rham operator can be now understood as particular cases of theorem~\ref{th:Gen_Weitzenboeck}, as we show in the following sections.


\subsection{Recovering the standard Weitzenböck indentity}
\label{sec:recoverstdWI}

The standard Weitzenböck identity deals with \so-covariant differential operators acting on a scalar $p$-form field~$\phi$.
In this case, it is possible to prove from eq.~(\ref{Eq_CoDerivS}) that
\begin{equation}
  \extd^{\dag} \phi =
  \mathring{\extD}^{\ddag} \phi =
  -\mathrm{I}_{a} \mathring{\extD} \mathrm{I}^{a}\phi,
  \label{Eq_Coderiv_Standard}
\end{equation}
and therefore the left-hand side of eq.~(\ref{Eq_Weitzenboeck_Generalisimo})
reduces to the standard de~Rham operator
\begin{equation}
  \Box_{\text{dR}} \phi =
  \left( \extd^{\dag} \extd + \extd \extd^{\dag} \right) \phi =
  \left(
    \mathring{\extD}^{\ddag} \mathring{\extD} +
    \mathring{\extD} \mathring{\extD}^{\ddag}
  \right) \phi.
  \label{Eq_Standard_dR}
\end{equation}
We remark that eqs.~(\ref{Eq_Coderiv_Standard})--(\ref{Eq_Standard_dR})
make use of the torsionless \so-covariant exterior derivative and generalized coderivative \emph{regardless} of whether manifold~$M$ is endowed with a nonvanishing torsion or not.
This may seem strange, but it is not.
Since (cf.~def.~\ref{def:ddag}) the exterior coderivative $\extd^{\dag}$ has some information on the metric (via the Hodge star operator) but none on torsion, torsion must be thoroughly absent from the right-hand side of eq.~(\ref{Eq_Coderiv_Standard}).
Therefore, regardless of the value of torsion on $M$, the generalized Beltrami operator $\blacksquare_{\text{B}}$ must be written in terms of the torsionless generalized Lie derivative (cf.~def.~\ref{def:D_Lie}) $\mathring{\mathcal{D}}_{a}$ as
$\blacksquare_{\text{B}} = -\mathring{\mathcal{D}}_{a} \mathring{\mathcal{D}}^{a}$.
From eq.~(\ref{Eq_LieD=Nabla+T}) we have that, for the torsionless connection, $\mathring{\mathcal{D}}_{a}$ and $\mathring{\nabla}_{\mu}$ are the same, and therefore $\blacksquare_{\text{B}}$ is just the standard Beltrami operator $\Box_{\text{B}}$.
Therefore, the standard Weitzenböck identity corresponds to
\begin{equation}
  \Box_{\text{dR}} \phi = \Box_{\text{B}} \phi +
  \mathrm{I}_{a} \mathring{\extD}^{2} \mathrm{I}^{a} \phi.
  \label{Eq_Weitzenbock_dR_ID2I}
\end{equation}
Using the Bianchi identity for $\mathring{\extD}^{2}$ and the Leibniz rule for $\mathrm{I}_{a}$, we have the Weitzenböck identity for the de~Rham operator,
\begin{equation}
  \Box_{\text{dR}} \phi = \Box_{\text{B}} \phi +
  \mathrm{I}_{a} \mathring{R}^{ab} \wedge \mathrm{I}_{b} \phi -
  \mathring{R}^{ab} \wedge \mathrm{I}_{ab} \phi.
  \label{Eq_Weitzenboeck_dR_Standard}
\end{equation}


\subsection{Recovering the Weitzenböck identity for the Lichnerowicz--de~Rham operator}
\label{sec:recoverLdRWI}

The Weitzenböck identity for the Lichnerowicz--de~Rham operator on torsionless geometries is recovered in a similar way.
Again, we have to pick $G=\SO$, but this time we also need to impose the torsionless condition on the geometry.
In the Lichnerowicz--de~Rham case, instead of a scalar $p$-form $\phi$ we have to
act on a $p$-form $\psi^{A}$ in some representation of \so.

From eq.~(\ref{Eq_CoDerivS}) we have that
\begin{equation}
  \mathring{\extD}^{\dag} \psi^{A} =
  -\left( -1 \right)^{d \left( p+1 \right) + \eta_{-}}
  \hodge \left( \mathring{\extD} \hodge \psi^{A} \right) =
  -\mathrm{I}_{a} \mathring{\extD} \mathrm{I}^{a} \psi^{A} =
  \mathring{\extD}^{\ddag} \psi^{A},
\end{equation}
and again the generalized Beltrami operator
$\blacksquare_{\text{B}} = -\mathring{\mathcal{D}}_{a} \mathring{\mathcal{D}}^{a}$
coincides with the standard Beltrami operator $\Box_{\text{B}}$.

With these considerations, theorem~\ref{th:Gen_Weitzenboeck} reduces in this case to
\begin{equation}
  \left(
    \mathring{\extD}^{\dag} \mathring{\extD} +
    \mathring{\extD} \mathring{\extD}^{\dag}
  \right)
  \psi^{A} =
  \Box_{\text{B}} \psi^{A} +
  \mathrm{I}_{a} \mathring{\extD}^{2} \mathrm{I}^{a} \psi^{A}.
  \label{Eq_Weitzenbock_LdR_ID2I}
\end{equation}
The Bianchi identity on $\mathring{\extD}^{2}$ and the Leibniz rule on
$\mathrm{I}_{a}$ lead us to the explicit expression
\begin{align}
  \Box_{\text{LdR}} \psi^{A} & =
  \Box_{\text{B}} \psi^{A} +
  \mathrm{I}_{a} \left(
    \nwse{\mathring{R}}{a}{b} \wedge \mathrm{I}^{b} \psi^{A} +
    \frac{1}{2} \mathring{R}^{cd} \rep{J_{cd}}{A}{B} \wedge \mathrm{I}^{a} \psi^{B}
  \right)
  \nonumber \\ & =
  \Box_{\text{B}} \psi^{A} +
  \mathrm{I}_{a} \nwse{\mathring{R}}{a}{b} \wedge \mathrm{I}^{b} \psi^{A} -
  \mathring{R}^{ab} \wedge \mathrm{I}_{ab} \psi^{A} +
  \frac{1}{2} \mathrm{I}_{a} \mathring{R}^{cd}
  \rep{J_{cd}}{A}{B} \wedge \mathrm{I}^{a} \psi^{B}.
  \label{Eq_Weitzenboeck_LdR_Standard}
\end{align}

It is straightforward to check that eqs.~(\ref{Eq_Weitzenboeck_dR_Standard})
and~(\ref{Eq_Weitzenboeck_LdR_Standard}) match the previously known
formulas~\cite{choquet1977analysis,nla.cat-vn2659416,doi:10.1142/S0218271803003785}.
Even if it represents the same well-known result, the clarity of expressions~(\ref{Eq_Weitzenbock_dR_ID2I}) and~(\ref{Eq_Weitzenbock_LdR_ID2I})
is noticeable when compared with the standard formulas found in the literature.

However, the real utility of theorem~\ref{th:Gen_Weitzenboeck} lies in its ability to extend the Weitzenböck identities to RC geometry with nonvanishing torsion and to general (super) algebras.


\subsection{The extended Weitzenböck identity for the Lichnerowicz--de~Rham operator with nonvanishing torsion}
\label{sec:extWISO}

As shown in section~\ref{sec:recoverstdWI}, the standard Weitzenböck identity for the de~Rham operator remains unchanged even on an RC geometry with nonvanishing torsion.
The Weitzenböck identity for the Lichnerowicz--de~Rham operator (see section~\ref{sec:recoverLdRWI} for the torsionless case), on the other hand, gets modified in a subtle way when theorem~\ref{th:Gen_Weitzenboeck} is applied to the full torsional case, where
\begin{equation}
  \blacksquare_{\text{LdR}} = \extD^{\ddag} \extD + \extD \extD^{\ddag}.
  \label{Eq_LLdR_SO}
\end{equation}
The torsionless Riemann curvature $\mathring{R}^{ab}$ morphes into the full Lorentz curvature~(\ref{Eq_Lorentz_Curvature}), and the Beltrami operator changes from the standard torsionless case $\Box_{\text{B}}$ to
\begin{equation}
  \blacksquare_{\text{B}} = -\mathcal{D}^{a} \mathcal{D}_{a}.
\end{equation}
Since $\mathcal{D}_{a}$ is the generalization (\ref{Eq_LieD=Nabla+T}) of $\nabla = \partial + \Gamma$, it is clear that $\blacksquare_{\text{B}}$ is also a generalization of the standard Beltrami operator. The generalized Weitzenböck identity corresponds to
\begin{equation}
  \blacksquare_{\text{LdR}} \psi^{A} =
  \blacksquare_{\text{B}} \psi^{A} +
  \mathrm{I}_{a} \nwse{R}{a}{b} \wedge \mathrm{I}^{b} \psi^{A} -
  R^{ab} \wedge \mathrm{I}_{ab} \psi^{A} +
  \frac{1}{2} \mathrm{I}_{a} R^{cd} \rep{J_{cd}}{A}{B}
  \wedge \mathrm{I}^{a} \psi^{B}.
\end{equation}


\subsection{The extended Weitzenböck identity for the generalized Lichnerowicz--de~Rham operator for arbitrary (super) groups}
\label{sec:extWIG}

In the general case, we concern ourselves with an arbitrary (super) group $G \supset \SO$ and differential operators acting on a $p$-form field $\Psi^{M}$ in some matrix representation of $G$.
Theorem~\ref{th:Gen_Weitzenboeck} and the Bianchi identity lead us to the general expression
\begin{equation}
  \blacksquare_{\text{LdR}} \Psi^{M} =
  \blacksquare_{\text{B}} \Psi^{M} +
  \mathrm{I}_{a} \left(
    \nwse{F}{a}{b} \wedge \mathrm{I}^{b} \Psi^{M} +
    \nwse{F}{M}{N} \wedge \mathrm{I}^{a} \Psi^{N}
  \right),
\end{equation}
where $\blacksquare_{\text{B}}$ stands for the generalized Beltrami operator introduced in definition~\ref{def:Beltrami_op} in terms of the generalized Lie derivative $\mathfrak{D}_{a}$ (see definition~\ref{def:D_Lie}), $\nwse{F}{M}{N}$ are the matrix components of the gauge curvature two-form~$F$, and $\nwse{F}{a}{b}$ are the matrix components of~$F$ restricted to the \so\ subalgebra.


\section{Open superalgebra of differential operators for the case of the indefinite special orthogonal group}
\label{sec:opensuper}

The most important case in the context of gravity corresponds to $G=\SO$.
In practical GW calculations, it proves valuable to observe that the operators
\begin{align}
  \mathrm{I}_{a} & :
  \Omega^{p} \left( M \right) \to \Omega^{p-1} \left( M \right), \\
  \extD & :
  \Omega^{p} \left( M \right) \to \Omega^{p+1} \left( M \right), \\
  \mathcal{D}_{a} & :
  \Omega^{p} \left( M \right) \to \Omega^{p} \left( M \right),
\end{align}
form an open superalgebra of differential operators
(all of them satisfy the Leibniz rule).
In this structure, $\mathrm{I}_{a}$ and $\extD$ play the
role of odd or ``fermionic'' operators and $\mathcal{D}_{a}$ acts as an even or ``bosonic'' operator.
These operators satisfy a super-Jacobi identity, where the tensor components of curvature and torsion correspond to the structure parameters.

Note first that the commutators of $\mathrm{I}_{a}$ and $\extD$ with $\mathcal{D}_{a}$ can be written as
\begin{align}
  \left[ \mathrm{I}_{a}, \mathcal{D}_{b} \right] & =
  -\nwse{T}{c}{ab} \mathrm{I}_{c},
  \label{Eq_Super[Ia,Db]} \\
  \left[ \extD, \mathcal{D}_{b} \right] & =
  \extD^{2} \mathrm{I}_{b} - \mathrm{I}_{b} \extD^{2}.
  \label{Eq_Super[D,Db]}
\end{align}
Using the super Jacobi identity
\begin{equation}
  \left\{ \extD, \left[ \mathrm{I}_{b}, \mathcal{D}_{a} \right] \right\} +
  \left[ \mathcal{D}_{a}, \left\{ \extD, \mathrm{I}_{b} \right\} \right] -
  \left\{ \mathrm{I}_{b}, \left[ \mathcal{D}_{a}, \extD \right] \right\} = 0,
\end{equation}
and since
\begin{equation}
  \left[ \mathcal{D}_{a}, \left\{ \extD, \mathrm{I}_{b} \right\} \right] =
  \left[ \mathcal{D}_{a}, \mathcal{D}_{b} \right],
\end{equation}
it is straightforward to prove that
\begin{align}
  \left[ \mathcal{D}_{a}, \mathcal{D}_{b} \right] & =
  \mathrm{I}_{ab} \extD^{2} + \extD^{2} \mathrm{I}_{ab} +
  \mathrm{I}_{a} \extD^{2} \mathrm{I}_{b} -
  \mathrm{I}_{b} \extD^{2} \mathrm{I}_{a} -
  \left(
    \extD \nwse{T}{c}{ab} \wedge \mathrm{I}_{c} + \nwse{T}{c}{ab} \mathcal{D}_{c}
  \right).
  \label{Eq_Super[Da,Db]}
\end{align}
Eq.~(\ref{Eq_Super[Da,Db]}) is the generalization of eq.~(\pageref{Eq_Orig_Nabla2}) from $\nabla_{\mu}$ to $\mathcal{D}_{a}$ and to arbitrary differential forms instead of $\vec{V}$.
In this way, the full superalgebra of differential operators, spanned by
$\mathrm{I}_{a}$, \extD, and $\mathcal{D}_{a}$,
can be written as
\begin{align}
  \left\{ \mathrm{I}_{a}, \extD \right\} & =
  \mathcal{D}_{a},
  \label{Eq_Super(I,D)} \\
  \left\{ \mathrm{I}_{a}, \mathrm{I}_{b} \right\} & = 0,
  \label{Eq_Super(I,I)} \\
  \left\{ \extD, \extD \right\} & = 2 \extD^{2},
  \label{Eq_Super(D,D)} \\
  \left[ \mathrm{I}_{a}, \mathcal{D}_{b} \right] & =
  -\nwse{T}{c}{ab} \mathrm{I}_{c},
  \\
  \left[ \extD, \mathcal{D}_{b} \right] & =
  \extD^{2} \mathrm{I}_{b} - \mathrm{I}_{b} \extD^{2},
  \\
  \left[ \mathcal{D}_{a}, \mathcal{D}_{b} \right] & =
  \mathrm{I}_{ab} \extD^{2} + \extD^{2} \mathrm{I}_{ab} +
  \mathrm{I}_{a} \extD^{2} \mathrm{I}_{b} -
  \mathrm{I}_{b} \extD^{2} \mathrm{I}_{a} -
  \left(
    \extD \nwse{T}{c}{ab} \wedge \mathrm{I}_{c} +
    \nwse{T}{c}{ab} \mathcal{D}_{c}
  \right),
\end{align}
where $\extD^{2}$ gives rise to terms proportional to the Lorentz curvature two-form $R^{ab}$ by virtue of the Bianchi identities.

This structure proves to be crucial to understanding the propagation of GWs on an RC geometry with nonvanishing torsion~\cite{Barrientos:2019awg}.
However, this article has a more general purpose.
In section~\ref{Sec_EikonalTorsion}, we study the most general case of the eikonal propagation of an arbitrary massless field $\Psi^{A}$ transforming in some matrix representation of the (super) algebra, i.e., when $\Psi^{A}$ satisfies the generalized wave equation
\begin{equation}
  \blacksquare_{\text{LdR}} \Psi^{A} =
  \left( \mathbb{D}^{\ddag} \mathbb{D} + \mathbb{DD}^{\ddag} \right) \Psi^{A} = 0.
\end{equation}


\section{The eikonal limit of the generalized Lichnerowicz--de~Rham operator for arbitrary (super) groups}
\label{Sec_EikonalTorsion}

In the preceding sections, we have presented several wave operators and the relationships between them. From this, it is clear that torsion enters the game in a nontrivial way. When torsion does not vanish, different kinds of fields may obey different wave operators on the same geometry. The difference is so vital that they can even lead to different geometric optics limits. Before treating the problem in rigorous terms, let us see a concrete example.

Let us consider a Riemann-Cartan spacetime with a nonvanishing torsion and an electromagnetic wave on it. The electromagnetic field $F=\mathrm{d}A$ satisfies the standard Maxwell equations in a vacuum, $\mathrm{d}^{\dag}F=0$, and $\mathrm{d}F=0$. Therefore, regardless of the torsional background, it will satisfy the wave equation
\begin{equation}
    \square_{\mathrm{dR}}F=0
\end{equation}
in terms of the canonical de Rham operator $\square_{\mathrm{dR}}=\mathrm{dd}^{\dag}+\mathrm{d}^{\dag}\mathrm{d}$, and the standard torsionless Weitzenb\"{o}ck relationship. The geometric optics limit is well-known, satisfying some equations we will review (\ref{Eq_P_Parallel_Transported},\ref{Eq_J_conserved}) to propagate amplitude and polarization.

On the same torsional background, let us consider the propagation of a gravitational wave. As we shall see, the phenomenology shares some similitudes, but it also has some significant differences.

A gravitational wave on a torsional background corresponds to independent perturbations on the vielbein $e^{a}$ and in the spin connection $\omega^{ab}$,
\begin{align}
  e^{a}  \mapsto \bar{e}^{a}&=e^{a}+\frac{1}{2}H^{a},\\
  \omega^{ab} \mapsto \bar{\omega}^{ab}&=\omega^{ab}+U^{ab}\left(  H,\partial H\right)  +V^{ab},
\end{align}
with
\begin{equation}
U^{ab}=-\frac{1}{2}\left(  \mathrm{I}_{a}\mathrm{D}H_{b}-\mathrm{I}_{b}\mathrm{D}H_{a}\right)  +\mathcal{O}\left(  H^{2}\right)  .
\end{equation}

Here $H^{a}=H^{a}{}_{b}e^{b}$ represents a metric perturbation and $V^{ab}$ and independent contorsional perturbation.
Ref.~\cite{Izaurieta:2019dix} thoroughly examined the kinematics of this system.

The Einstein--Hilbert term gives rise to an expression~\cite{Barrientos:2019awg} of the kind
\begin{equation}
\blacksquare_{\mathrm{LdR}}H^{a}+\mathrm{subleading~terms}=0.\label{Eq_H_generic}
\end{equation}

The generalized de Rham operator $\blacksquare_{\mathrm{LdR}}$ leads to the canonic dispersion relation and the correct speed for gravitational waves. However, its subleading behavior provides an anomalous propagation of amplitude and polarization. Various (admissible) theories provide diverse subleading terms in eq.~(\ref{Eq_H_generic}), changing the right-hand side of eqs.~\ref{Eq_J_GW} and~\ref{Eq_P_GW} differently. The crucial point is that we have different wave operators for different kinds of fields on the same geometry.


To characterize the wave operators, we must study its eikonal or geometric optics limit.
In quantum mechanics it is known as the WKB approximation, and in the context of Riemannian gravity it is fundamental to understand and define GWs~\cite{Misner1973Gravitation,Maggiore:1900zz}.
When the eikonal conditions are not fullfilled, sometimes the separation
wave/background cannot even be properly defined (as in the case of GWs~\cite{Maggiore:1900zz}) or it may lead to seemingly anomalous behaviors.
For instance, when the eikonal approximation conditions are not fulfilled for EMWs in standard general relativity, it is even possible to find (in vacuum!) EMWs not traveling at the speed of light because of anomalous dispersion relations~\cite{PhysRevD.96.044021,Asenjo:2016hnc}.

That is why in the current work we focus our attention on the propagation of a $p$-form field $\Psi^{A}$ in some representation of the (super) algebra $\mathfrak{g}$ satisfying the generalized wave equation (cf.~def.~\ref{def:Lichnerowicz_op})
\begin{equation}
  \blacksquare_{\text{LdR}} \Psi^{A} = \left(
    \mathbb{D}^{\ddag} \mathbb{D} + \mathbb{DD}^{\ddag}
  \right) \Psi^{A} = 0,
  \label{Eq_Gen_Wave_Eq}
\end{equation}
obeying the well-known two-parameter eikonal approximation throughout~\cite{Maggiore:1900zz,Misner1973Gravitation}. 

We begin by parameterizing the \emph{wave} $\Psi^{A}$ in terms of a real \emph{phase} $\theta$ and a complex-valued $p$-form \emph{amplitude} $\psi^{A}$ as
\begin{equation}
  \Psi^{A} = \euler^{i\theta} \psi^{A}.
  \label{Eq_Eikonal_start}
\end{equation}
We shall assume that $\psi^{A}$ changes slowly while $\theta$ changes rapidly on a characteristic length scale $\lambdabar$.
The wave propagates on a background manifold $M$ that changes slowly on a characteristic length scale at least as large as $L$.
In terms of these two characteristic lenghts, the \emph{eikonal parameter} $\varepsilon$ is defined as
\begin{equation}
  \varepsilon=\frac{\lambdabar}{L}.
\end{equation}
The eikonal limit corresponds to the case $\varepsilon \ll1$.

In the eikonal limit, we can expand the amplitude in powers of $\varepsilon$ as
\begin{equation}
  \psi^{A} = \sum_{n=0}^{\infty} \psi_{\left( n \right)}^{A},
  \label{Eq_Amplitude_Decomposition}%
\end{equation}
where $\psi_{\left( n \right)}^{A}$ is a term of order $\varepsilon^{n}$.
Eq.~(\ref{Eq_Amplitude_Decomposition}) allows us to split $\psi^{A}$ as a dominant, $\lambdabar$-independent $p$-form amplitude $\psi_{\left( 0 \right)}^{A}$ plus small ``geometric optics deviations''
$\psi_{\left( 1 \right)}^{A}$,
$\psi_{\left( 2 \right)}^{A}$, \ldots\
due to the finite wavelength.

Let $k = \extd \theta = k_{a} e^{a} = k_{\mu} \extd x^{\mu}$
be the \emph{wave one-form}.
Using eq.~(\ref{Eq_Eikonal_start}) and Theorem~\ref{th:Gen_Weitzenboeck}, it is possible to show that
\begin{equation}
  \blacksquare_{\text{LdR}} \Psi^{A} =
  \euler^{i\theta} \left[
    k^{2} \psi^{A} - 2i \left(
      k^{a} \mathfrak{D}_{a} \psi^{A} -
      \frac{1}{2} \psi^{A} \mathbb{D}^{\ddag} k
    \right) +
    \blacksquare_{\text{LdR}} \psi^{A}
  \right],
  \label{Eq_Wave_fase_amp}
\end{equation}
where $\mathfrak{D}_{a}$ stands for the generalized Lie derivative introduced in definition~\ref{def:D_Lie}, and $k^{2}: = \left(-1\right)^{\eta_{-}}\left(k\wedge\ast k\right)= k_{a} k^{a} = k_{\mu} k^{\mu}$.
Plugging the expansion~(\ref{Eq_Amplitude_Decomposition}) for the amplitude $\psi^{A}$ into eq.~(\ref{Eq_Wave_fase_amp}) and demanding that the generalized wave equation~(\ref{Eq_Gen_Wave_Eq}) be satisfied at all orders in $\varepsilon$, we find
\begin{itemize}
\item Order $\varepsilon^{-2}$:
\begin{equation}
  k^{2} \psi_{\left( 0 \right)}^{A} = 0.
  \label{Eq_E-2}
\end{equation}

\item Order $\varepsilon^{-1}$:
\begin{equation}
  k^{2} \psi_{\left( 1 \right)}^{A} - 2i \left(
    k^{a} \mathfrak{D}_{a} \psi_{\left( 0 \right)}^{A} -
    \frac{1}{2} \psi_{\left( 0 \right)}^{A} \mathbb{D}^{\ddag} k
  \right) = 0.
  \label{Eq_E-1}
\end{equation}

\item Order $\varepsilon^{0}$ and higher:
\begin{align}
  \sum_{n=0}^{\infty} \left[
    k^{2} \psi_{\left( n+2 \right)}^{A} - 2i \left(
      k^{a} \mathfrak{D}_{a} \psi_{\left( n+1 \right)}^{A} -
      \frac{1}{2} \psi_{\left( n+1 \right)}^{A} \mathbb{D}^{\ddag} k
    \right) +
    \blacksquare_{\text{LdR}} \psi_{\left( n \right)}^{A}
  \right] = 0.
  \label{Eq_E-0}
\end{align}
\end{itemize}

In the power series expansion~(\ref{Eq_Amplitude_Decomposition}),
$\psi_{\left( 0\right)}^{A}$ corresponds to the $\lambdabar$-independent leading term, and therefore $\psi_{\left(  0\right)  }^{A}\neq0$.
This is what we have learned from eqs.~(\ref{Eq_E-2})--(\ref{Eq_E-0}):
\begin{itemize}
\item At leading order, the dispersion relation,
\begin{equation}
  k^{2}=k_{a} k^{a} = k_{\mu} k^{\mu} = 0.
  \label{Eq_Leading}
\end{equation}

\item At subleading order, the propagation of the amplitude and polarization,
\begin{equation}
  k^{a} \mathfrak{D}_{a} \psi_{\left( 0 \right)}^{A} -
  \frac{1}{2} \psi_{\left( 0 \right)}^{A} \mathbb{D}^{\ddag} k = 0.
  \label{Eq_Subleading}
\end{equation}

\item Higher-order deviations from geometric optics,
\begin{equation}
  \sum_{n=0}^{\infty} \left[
    \blacksquare_{\text{LdR}} \psi_{\left( n \right)}^{A} - 2i \left(
      k^{a} \mathfrak{D}_{a} \psi_{\left( n+1 \right)}^{A} -
      \frac{1}{2} \psi_{\left( n+1 \right)}^{A} \mathbb{D}^{\ddag} k
    \right)
  \right] = 0.
  \label{Eq_Higher_Order_Deviations}
\end{equation}
\end{itemize}

A few comments are in order.

First, neither torsion nor the (super) group symmetry $G$ affect the dispersion relation [cf.~eq.~(\ref{Eq_Leading})], $k^{2} = 0$.
The wave one-form $k$ is null regardless of them.
From this and the fact that $\extd k = \extd^{2} \theta = 0$,
it is straightforward to prove that regardless of whether the background geometry has torsion or not, one always has
\begin{equation}
  k^{\mu} \mathring{\nabla}_{\mu} k^{\nu} = 0,
\end{equation}
i.e., $k=k_{\mu}\extd x^{\mu}$ moves along null \emph{geodesics} and
not along null self-parallels.
The dispersion relation does not care at all whether the background geometry has torsion or not.
This result is key in the context of GW observations such as GW170817/GRB170817A~\cite{Barrientos:2019awg}.
A GW obeying the generalized wave equation~(\ref{Eq_Gen_Wave_Eq}) in the eikonal limit will travel at the speed of light on torsionless null geodesics regardless of torsion.

The same is not true regarding the propagation of amplitude and polarization.
To analyze their behavior at subleading order,
let us write the leading, complex-valued $p$-form $\psi_{\left( 0 \right)}^{A}$
as the product of a real scalar amplitude $\varphi$
and a complex-valued \emph{polarization} $p$-form $P^{A}$,
\begin{equation}
  \psi_{\left( 0 \right)}^{A} = \varphi P^{A}.
  \label{Eq_scalar_polarization}
\end{equation}
We choose the following normalization for the polarization:%
\footnote{Here, $C_{AB}$ is a constant bilinear invariant tensor of the gauge (super) group $G$, such as the Killing metric. If an arbitrary field $\psi^{A}$ transforms as $\psi^{A} \mapsto \psi'^{A} = G^{A}{}_{B} \psi^{B}$, then the squared magnitude of $\psi^{A}$, $C_{AB} \bar{\psi}^{A} \wedge \psi^{B}$, where $\bar{\psi}^{A}$ denotes the complex conjugate of $\psi^{A}$, remains invariant under the action of the group. Equivalently, we may write $\mathbb{D} C_{AB} = 0$.}
\begin{equation}
  \left( -1 \right)^{\eta_{-}} \hodge
  \left( C_{AB} \bar{P}^{A} \wedge \hodge P^{B} \right) = 1,
\end{equation}
whence
\begin{equation}
  \varphi^{2} = \left( -1 \right)^{\eta_{-}} \hodge \left(
    C_{AB} \bar{\psi}_{\left( 0 \right)}^{A}
    \wedge \hodge \psi_{\left( 0 \right)}^{B}
  \right).
  \label{Eq_scalar2}
\end{equation}

On a torsionless Riemannian geometry with $G=\SO$,
it is possible to prove that in the eikonal limit the polarization satisfies
\begin{equation}
  k^{\lambda} \mathring{\nabla}_{\lambda}
  \nwse{P}{A}{\mu_{1} \cdots \mu_{p}} = 0,
  \label{Eq_P_Parallel_Transported}
\end{equation}
i.e., the polarization $p$-form $P^{A}$ is parallel transported along the torsionless null geodesic described by $k^{\lambda}$.
In the torsionless case, the real amplitude $\varphi$ and the wave one-form $k=k_{\mu}\extd^{\mu}$ define the ray-current density one-form $J = \varphi^{2} k$ (called ``number of photons'' current density in the electromagnetic case), which is conserved, i.e., satifies the continuity equation
\begin{equation}
  \extd^{\dag} J = 0.
  \label{Eq_J_conserved}
\end{equation}
Eqs.~(\ref{Eq_P_Parallel_Transported}) and~(\ref{Eq_J_conserved}) describe a very general behavior for any physical interaction progating at the speed of light on a curved torsionless Riemannian background.
For instance, EMWs and GWs in standard GR obey these relations~\cite{Maggiore:1900zz}.

Torsion and the additional symmetries can change the relations (\ref{Eq_P_Parallel_Transported}) and (\ref{Eq_J_conserved}) in a nontrivial way. When a field satisfies the wave equation in terms of the standard de Rham operator (e.g., an EMW in a vacuum, $\left(\extd\extd^{\dag}+\extd^{\dag}\extd\right)F=0$ ) on a Riemann-Cartan manifold with nonvanishing torsion, it satisfies equations (\ref{Eq_P_Parallel_Transported}) and (\ref{Eq_J_conserved}). In this case, the presence of torsion is irrelevant for the wave propagation.
However, the same is not valid for a field $\Psi^{A}$ on a representation of $G$. For these fields, amplitude and polarization propagate in an anomalous way.

In fact, using eq.~(\ref{Eq_Subleading}) it is possible to prove (see Appendix~\ref{Appendix-Amplitude} for the details) that
\begin{equation}
  \mathrm{d}^{\dag}J=\mathbb{T}_{abc}\left(  \eta^{ab}-\Pi^{ab}\right)  J^{c},
  \label{Ec_non_conservation}
\end{equation}
where $\mathbb{T}_{abc}$ are the orthonormal-basis components of the $G$-torsion and $\Pi^{ab}$ corresponds to
\begin{equation}
  \Pi^{ab}=\left(  -1\right)  ^{\eta_{-}}C_{AB}\ast \left[  \mathrm{I}^{a}\bar{P}^{A}\wedge \ast \mathrm{I}^{b}P^{B}+\mathrm{I}^{b}\bar{P}^{A}\wedge\ast \mathrm{I}^{a}P^{B}\right]
\end{equation}
Eq.~(\ref{Ec_non_conservation}) makes clear that when torsion does not vanish, the continuity equation for the $\psi$-ray density $J$ is broken, $\extD^{\dag} J \neq 0$.
Since torsion does not affect the dispersion relation, frequencies remain unchanged.
Therefore, what is affected by eq.~(\ref{Ec_non_conservation}) is the propagation of the amplitude $\varphi$.
Depending on the sign of the right-hand side of eq.~(\ref{Ec_non_conservation}), torsion can either reinforce or damp the propagation of the wave.
The right-hand side of eq.~(\ref{Ec_non_conservation}) is proportional to $k^{a}$, and therefore the effect is stronger at higher frequencies.

The propagation of polarization behaves similarly.
Using again eq.~(\ref{Eq_Subleading}), it is possible to prove (see Appendix~\ref{Appendix-Polarization} for the details) that
\begin{align}
  k^{a}\mathcal{\mathring{D}}_{a}P^{A}  & =\frac{1}{2}\left \{  -\mathbb{T}_{abc}\Pi^{ab}P^{A}+\left(  \mathbb{T}_{bac}+\mathbb{T}_{abc}-\mathbb{T}_{cab}\right)  e^{a}\wedge \mathrm{I}^{b}P^{A}+\right.  \nonumber \\
  & \left.  -\left[  \frac{1}{2}\left(  T_{bac}-T_{abc}+T_{cab}\right)  \left[J^{ab}\right]  ^{A}{}_{B}+2\mathrm{I}_{c}a^{A}{}_{B}\right]  P^{B}\right \}k^{c},
  \label{Eq_Propag_Polarization_Gen}
\end{align}
where $\mathring{\mathcal{D}}_{a}$ stands for the \so-covariant torsionless generalized Lie derivative introduced in definition~\ref{def:D_Lie}.
According to Lemma~\ref{lem:calD=nabla}, $\mathring{\mathcal{D}}_{a}$ behaves exactly as the torsionless covariant derivative of Riemannian geometry, $\mathring{\nabla}_{\mu}$.
Therefore, $k^{a} \mathring{\mathcal{D}}_{a} P^{A} = k^{\mu} \mathring{\nabla}_{\mu} P^{A} = 0$ corresponds to the polarization $p$-form $P^{A}$ being parallel transported by the wave as it propagates along the null torsionless geodesic.
The nonvanishing of the right-hand side of eq.~(\ref{Eq_Propag_Polarization_Gen}) represents a departure from this behavior.
Since it depends on the components of $G$-torsion $\mathbb{T}_{abc}$ and
the ones of the \SO\ torsion $T_{abc}$, this means that while the wave propagates along the null torsionless geodesic, the background torsion interacts with the polarization components, scrambling them along the way.

To make the content of eqs.(\ref{Ec_non_conservation}) and (\ref{Eq_Propag_Polarization_Gen}) more transparent, let us write them in the standard coordinate basis for the metric mode of a gravitational wave with dominant term $H_{\left(  0\right)  }^{a}=e^{a}{}_{\mu}e^{i\theta}P^{\mu}{}_{\nu}\mathrm{d}x^{\nu}$ and polarization $P^{\mu}{}_{\nu}$. For this field, the anomalous propagation of the amplitude current $J^{\lambda}=\varphi^{2}k^{\lambda}$ and polarization corresponds to
\begin{equation}
    \mathring{\nabla}_{\lambda}J^{\lambda}=T_{\mu \nu \lambda}\left(  \Pi^{\mu \nu}-g^{\mu \nu}\right)  J^{\lambda},\label{Eq_J_GW}
\end{equation}
and
\begin{equation}
  k^{\lambda}\mathcal{\mathring{\nabla}}_{\lambda}P_{\mu \nu}=-\frac{1}{2}k^{\gamma}\left[  P_{\mu \nu}\Pi^{\alpha \beta}T_{\alpha \beta \gamma}+\left(  -T_{\lambda \gamma \mu}+T_{\mu \gamma \lambda}-T_{\gamma \lambda \mu}\right)  P^{\lambda}{}_{\nu}+\left(  T_{\lambda \gamma \nu}+T_{\nu \gamma\lambda}-T_{\gamma \lambda \nu}\right)  P^{\lambda}{}_{\mu}\right]  ,\label{Eq_P_GW}
\end{equation}
with%
\begin{equation}
  \Pi^{\mu \nu}=\bar{P}^{\mu}{}_{\gamma}P^{\gamma \nu}+\bar{P}^{\nu}{}_{\gamma}P^{\gamma \mu}.
\end{equation}

Therefore, the anomalous propagation of amplitude and polarization of a GW depends crucially on the geometry's torsional background.

Eqs.~(\ref{Ec_non_conservation}) and~(\ref{Eq_Propag_Polarization_Gen}) imply a subtle but in principle detectable difference between fields obeying a wave equation given in terms of the usual de~Rham operator~(cf. def.~\ref{def:dRbox}) (e.g., an EW in vacuum) and a field obeying a wave equation in terms of the generalized Lichnerowicz--de~Rham operator~(cf.~def.~\ref{def:Lichnerowicz_op}) or its \SO-restricted version~(\ref{Eq_LLdR_SO}) (e.g., a GW in the eikonal limit~\cite{Barrientos:2019awg}).
Both waves follow torsionless null geodesics, meaning, in particular, that no ``torsional lensing'' should be expected in either case.
However, a wave of the first kind is completely oblivious to the torsional background.
Its ray-density current obeys the continuity equation, and its polarization is parallel transported along the path of the wave.
In the second case, measurements of the anomalous propagation of amplitude~(\ref{Ec_non_conservation}) and polarization~(\ref{Eq_Propag_Polarization_Gen}) could be used as probes to detect any background torsion.
In this way, careful measurements of amplitude and polarization of fields in some representation of \SO\ (such as fermionic fields or GWs) propagating through vast distances could be used in principle as a way to test for the presence of spacetime torsion.


\section{Conclusions}
\label{sec:final}

The current work has the objective of studying the propagation of waves on spaces with a torsional RC geometry, how torsion modifies this propagation, and how the presence of additional gauge fields can affect it.

To accomplish this, the article is divided in two parts.
The first one, more mathematical, focuses in studying in detail the underlying relation between the Lichnerowicz--de~Rham operator (``the mathematician's Laplacian'') and the Beltrami operator (``the physicist's Laplacian'') through the Weitzenböck identity.
When expressed in terms of tensor components
[cf.~eqs.~(\ref{Eq_Ej-W-1})--(\ref{Eq_Ej-W-3})],
the identity remains obscure~\cite{choquet1977analysis,nla.cat-vn2659416}.
However, using Theorem~\ref{th:Gen_Weitzenboeck}, the same identity is expressed in a simple way, clarifying the underlying structure, and hinting at how the Lichnerowicz--de~Rham operator should be generalized for the torsional case.
It is worth observing that the generalized form of the Lichnerowicz--de~Rham operator~(\ref{Eq_LLdR_SO}) is also the one that GWs obey when torsion does not vanish~\cite{Barrientos:2019awg}.

The second part of the article, more phenomenological, focuses on the eikonal propagation of waves obeying the generalized Lichnerowicz--de~Rham wave equation~(\ref{Eq_Gen_Wave_Eq}).
The results proven in this second part are very general, and apply to the propagation of any field in some representation of an algebra including \SO, such as GWs and fermionic fields obeying the wave equation.
Both the standard (cf.~def.~\ref{def:dRbox}) and the generalized version (cf.~def.~\ref{def:Lichnerowicz_op}) of the de~Rham operator lead to the same dispersion relation and predict that waves sholud propagate along the same torsionless null geodesics, regardless of the value of the background torsion.
From a phenomenological point of view, this is an important fact.
The multimessenger detection of GWs and EMWs in the event GW170817/GRB170817A indicates that despite being described by different kinds of wave operators, they must follow the same dispersion relation, speed, and kind of trajectory (i.e., null torsionless geodesics) regardless of the background torsion~\cite{Izaurieta:2019dix,Valdivia:2017sat}. %
Therefore, the observation GW170817/GRB170817A agrees with the current work, and it does not rule out a torsional background.
Torsion could be part of what we attribute to dark matter effects~\cite{Tilquin:2011bu,Alexander:2019wne,Magueijo:2019vmk,Barker:2020gcp,Alexander:2020umk,Izaurieta:2020xpk}; even more, a dark matter with a nonvanishing spin tensor creating torsion might solve the Hubble tension problem~\cite{Izaurieta:2020xpk}.
%

The difference between the standard and generalized operators arises when we consider how amplitude and polarization propagate.
For fields in some representation of an algebra including \SO, eq.~(\ref{Ec_non_conservation}) implies that the ray-density one-form $J = \varphi^{2} k$ is no longer conserved, and the background torsion can
damp or amplify the amplitude $\varphi$ while the wave propagates.
In the same way, from eq.~(\ref{Eq_Propag_Polarization_Gen}) we see that polarization is no longer parallel transported along the torsionless null geodesic, but the
polarization modes get scrambled along the way by the background torsion.

In principle, these two phenomena could enable a way to detect torsion through GWs.
In a recent paper~\cite{Barrientos:2019awg}, we discussed a gravitational version of the optical rotation known as Faraday effect, with torsion playing the role of the glass in a magnetic field.
Such an effect could be used as a probe for torsion detection.


\section*{Acknowledgments}

We are grateful to
Yuri Bonder,
Fabrizio Canfora,
Oscar Castillo-Felisola,
Fabrizio Cordonier-Tello,
Crist\'{o}bal Corral,
Nicol\'{a}s Gonz\'{a}lez,
Perla Medina,
Daniela Narbona,
Julio Oliva,
Francisca Ram\'{\i}rez,
Patricio Salgado,
Sebasti\'{a}n Salgado,
Jorge Zanelli,
and
Alfonso Zerwekh
for many enlightening conversations.
JB acknowledges financial support from CONICYT grant 21160784. JB also thanks the Institute of Mathematics of the Czech Academy of Sciences, where part of this work was carried out, for their warm hospitality.
FI acknowledges financial support from the Chilean government through Fondecyt grants 1150719 and 1180681. FI is thankful of the emotional support of the Netherlands Bach Society. They made freely available superb quality recordings of the music of Bach, and without them, this work would have been impossible.
OV acknowledges VRIIP UNAP for financial support through Project VRIIP0258-18.


\appendix


\section{Anomalous propagation of amplitude}
\label{Appendix-Amplitude}

In this Appendix we derive eq.~(\ref{Ec_non_conservation}).
Let us start by using eq.~(\ref{Eq_scalar2}) to write
\begin{equation}
  k^{a}\mathfrak{D}_{a}\varphi^{2}=\frac{1}{p!}C_{AB}\left(  k^{a}\mathfrak{D}_{a}\bar{\psi}_{\left(  0\right)  }^{Aa_{1}\cdots a_{p}}\psi_{\left(  0\right)  a_{1}\cdots a_{p}}^{B}+\bar{\psi}_{\left(  0\right)}^{Aa_{1}\cdots a_{p}}k^{a}\mathfrak{D}_{a}\psi_{\left(  0\right)  a_{1}\cdots a_{p}}^{B}\right).
\end{equation}

Using eq.~(\ref{Eq_Subleading}), it is possible to prove that

\begin{equation}
  \frac{1}{p!}e^{a_{1}}\wedge \cdots \wedge e^{a_{p}}k^{a}\mathfrak{D}_{a}\psi_{\left(  0\right)  a_{1}\cdots a_{p}}^{A}+k^{a}\mathrm{I}_{a}\mathbb{T}^{b}\wedge \mathrm{I}_{b}\psi_{\left(  0\right)  }^{A}+\frac{1}{2}\psi_{\left(  0\right)  }^{A}\mathfrak{D}_{a}k^{a}=0.
\end{equation}

Combining both relations, we get
\begin{align}
  k^{a}\mathfrak{D}_{a}\varphi^{2}  & =-\left(  -1\right)  ^{\eta_{-}}C_{AB}\ast \left \{  k^{a}\left[  \mathrm{I}_{a}\mathbb{T}^{b}\wedge \mathrm{I}_{b}\bar{\psi}_{\left(  0\right)  }^{A}\wedge \ast \psi_{\left(  0\right)  }^{B}+\bar{\psi}_{\left(  0\right)  }^{A}\wedge \ast \left(  \mathrm{I}_{a}\mathbb{T}^{b}\wedge \mathrm{I}_{b}\psi_{\left(  0\right)  }^{B}\right)\right]  +\right. \nonumber \\
  & \left.  +\bar{\psi}_{\left(  0\right)  }^{A}\wedge \ast \psi_{\left(0\right)  }^{B}\mathfrak{D}_{a}k^{a}\right \}
\end{align}

Replacing here eq.~(\ref{Eq_scalar2}), we get
\begin{equation}
  \mathfrak{D}_{a}J^{a}=\mathbb{T}_{abc}\Pi^{ab}J^{c},\label{Eq_DJ_App}
\end{equation}

where $\mathbb{T}_{abc}$ are the orthonormal-basis components of the $G$-torsion, $\mathbb{T}_{a} = \frac{1}{2} \mathbb{T}_{abc} e^{b} \wedge e^{c}$,
\begin{equation}
  J^{a}=\varphi^{2}k^{a},
\end{equation}
and
\begin{equation}
  \Pi^{ab}=\left(  -1\right)  ^{\eta_{-}}C_{AB}\ast \left[  \mathrm{I}^{a}\bar{P}^{A}\wedge \ast \mathrm{I}^{b}P^{B}+\mathrm{I}^{b}\bar{P}^{A}\wedge\ast \mathrm{I}^{a}P^{B}\right]  .
\end{equation}

Splitting the covariant derivative as in eq.~(\ref{Eq_G-contorsion}), we find
\begin{align}
  \mathfrak{D}_{a}J^{a}  & =-\mathbb{D}^{\ddag}J,\nonumber \\
  & =-\mathrm{\mathring{D}}^{\ddag}J+\mathrm{I}^{a}\mathbb{K}_{ab}J^{b},\nonumber \\
  & =-\mathrm{d}^{\dag}J+\mathbb{T}^{a}{}_{ab}J^{b}.
\end{align}

This relation lead us to
\begin{equation}
  \mathrm{d}^{\dag}J=\mathbb{T}_{abc}\left(  \eta^{ab}-\Pi^{ab}\right)  J^{c}.
\end{equation}
for the anomalous propagation of amplitude when $\mathbb{T}^{a}\neq 0$ and eq.~(\ref{Eq_Gen_Wave_Eq}).

\section{Anomalous propagation of polarization}
\label{Appendix-Polarization}
In this Appendix we derive eq.~(\ref{Eq_Propag_Polarization_Gen}).
Replacing eq.~(\ref{Eq_scalar_polarization}) in eq.~(\ref{Eq_Subleading}) we get
\begin{equation}
  k^{a}\mathfrak{D}_{a}P^{A}+\left(  \frac{1}{\varphi}k^{a}\mathfrak{D}_{a}\varphi+\frac{1}{2}\mathfrak{D}_{a}k^{a}\right)  P^{A}=0.
\end{equation}

Replacing here eq.~(\ref{Eq_DJ_App}) with $J^{a}=\varphi^{2}k^{a}$, we have
\begin{equation}
  \frac{1}{\varphi}k^{a}\mathfrak{D}_{a}\varphi+\frac{1}{2}\mathfrak{D}_{a}k^{a}=\frac{1}{2\varphi^{2}}\mathbb{T}_{abc}\Pi^{ab}J^{c},
\end{equation}
and therefore
\begin{equation}
  k^{a}\mathfrak{D}_{a}P^{A}=-\frac{1}{2}\mathbb{T}_{abc}\Pi^{ab}k^{c}P^{A}.\label{Eq_Prop_Polarization_Short}
\end{equation}

From this equation, it is already evident that the orthonormal-basis components of torsion, $\mathbb{T}_{abc}$, modify the propagation of the polarization $p$-form $P^{A}$.
To make this effect more explicit, let us split $\mathbb{D}$ (cf.~def.~\ref{def:covextders}) in terms of the contorsion and $G$-contorsion one-forms as
\begin{equation}
  \mathfrak{D}_{a} P^{A} =
  \mathring{\mathcal{D}}_{a} P^{A} + \mathrm{I}_{a} \left(
    \frac{1}{2} \kappa^{mn} \rep{J_{mn}}{A}{B} + \nwse{a}{A}{B}
  \right)
  \wedge P^{B} + \mathbb{K}_{ab} \wedge \mathrm{I}^{b} P^{A}.
\end{equation}
Writing the contorsion $\kappa^{mn}$ and $G$-contorsion $\mathbb{K}_{ab}$ one-forms in terms of the torsion and $G$-torsion as in eq.~(\ref{Eq_K=K(T)}), and replacing this in eq.~(\ref{Eq_Prop_Polarization_Short}), we finally find the general expression given in eq.~(\ref{Eq_Propag_Polarization_Gen}),
\begin{align}
	k^{a}\mathcal{\mathring{D}}_{a}P^{A} &  =\frac{1}{2}\left \{  -\mathbb{T}_{abc}\Pi^{ab}P^{A}+\left(  \mathbb{T}_{bac}+\mathbb{T}_{abc}-\mathbb{T}_{cab}\right)  e^{a}\wedge \mathrm{I}^{b}P^{A}+\right.  \nonumber \\
	&  \left.  -\left[  \frac{1}{2}\left(  T_{bac}-T_{abc}+T_{cab}\right)  \left[ J^{ab}\right]  ^{A}{}_{B}+2\mathrm{I}_{c}a^{A}{}_{B}\right]  P^{B}\right \} k^{c},\label{Eq_Propag_Polarization_Gen_App}
\end{align}
from which it is evident that torsion modifies the propagation of polarization along the trajectory.

\bibliographystyle{utphys}
\bibliography{Bib_12_jul_2020}

\end{document}